\PassOptionsToPackage{dvipsnames}{xcolor}
\documentclass[10pt]{article}
\usepackage{color}
\usepackage{amsmath,amsfonts,amssymb,amsthm,mathtools}
\usepackage{mhchem}
\usepackage{todonotes}
\usepackage{booktabs}
\usepackage{siunitx}
\usepackage{subcaption}
\usepackage{multicol}
\usepackage[prefix=sol-]{xcolor-solarized}
\newcommand{\mset}[1]{\left\{\!\!\left\{#1\right\}\!\!\right\}}

\DeclareMathOperator{\koenig}{K}
\newtheorem{definition}{Definition}
\newtheorem{lemma}[definition]{Lemma}

\newcommand{\supp}{\operatorname{supp}}

\newcommand{\AX}[1]{\textnormal{#1}}

\usepackage{mod}
\tikzset{eMorphism/.style={->, >=stealth}}
\tikzset{eIsomorphism/.style={double, double equal sign distance}}

\tikzset{vDot/.style={fill,minimum size=0, inner sep=1pt, circle, outer sep=0}}
\tikzset{eDot/.style={}}
\tikzset{eDotMorphism/.style={->, >=stealth, dashed, red}}
\tikzset{gDot/.style={node distance=1em and 1em}}
\tikzset{gCat/.style={node distance=2em and 2em}}
\tikzset{vCat/.style={draw, minimum size=2.5em}}

\usetikzlibrary{shapes,positioning,arrows,calc,matrix,shadows,decorations.markings}
\tikzset{hnodeNoDraw/.style={ellipse,minimum height=12, minimum width=17}}
\tikzset{hnode/.style={draw,hnodeNoDraw}}
\tikzset{hxnode/.style={hnode, minimum size=50}}
\tikzset{tnode/.style={draw,circle,fill=black,inner sep=0,minimum size=3,text height=0ex,text depth=0ex}}
\tikzset{hedgeNoDraw/.style={rectangle,minimum height=8, minimum width=8}}
\tikzset{hedge/.style={draw,hedgeNoDraw}}
\tikzset{edge/.style={->,>=stealth',thick}}
\tikzset{tedge/.style={->,>=stealth'}}
\newcommand{\multiedge}[4][20]{
\pgfmathsetmacro{\num}{#4-1}
\foreach \i in {0, ..., \num}{
	\pgfmathsetmacro{\angleSep}{40}
	\pgfmathsetmacro{\angle}{\angleSep*(\i-#4/2+.5)}
	\draw[edge] (#2) to [bend right=\angle] (#3);
}
}

\newcommand\ioEdgeDist{0.9}
\newcommand\isomerizationOffset{15}
\newcommand\makeIO[5][\ioEdgeDist]{{
	\newcommand{\Angle}{#3}
	\newcommand{\Anchor}{#2}
	\newcommand{\Dist}{#1}
	\newcommand{\InText}{#4}
	\newcommand{\OutText}{#5}
	\newcommand\Offset\isomerizationOffset
	\draw[edge] (\Anchor.\Angle+\Offset) to node[auto, every label]{\OutText} ($(\Anchor.\Angle+\Offset) + (\Angle:\Dist)$);
	\draw[edge] ($(\Anchor.\Angle-\Offset) + (\Angle:\Dist)$) to node[auto, every label]{\InText} (\Anchor.\Angle-\Offset);
}}

\tikzset{
	matrixInnerSep/.style={inner sep=0, nodes={inner sep=.3333em)}},
	every matrix/.style={ampersand replacement=\&},
	every node/.style={font=\footnotesize}
}

\newcommand\autocataExBegin{%
	\begin{tikzpicture}
	\def\labAB{$e_1$}
	\def\labBA{$e_2$}
	\def\labBC{$e_3$}
	\def\labCAB{$e_4$}
	\def\labED{$e_5$}
	\def\labAI{$e_A^-$}\def\labAO{$e_A^+$}
	\def\labBI{$e_B^-$}\def\labBO{$e_B^+$}
	\def\labCI{$e_C^-$}\def\labCO{$e_C^+$}
	\def\labDI{$e_D^-$}\def\labDO{$e_D^+$}
	\def\labEI{$e_E^-$}\def\labEO{$e_E^+$}
	\tikzset{
		styleA/.style={},
		styleB/.style={},
		styleC/.style={},
		styleAB/.style={},	styleBA/.style={},
		styleBC/.style={},	styleBBC/.style={},	styleBCB/.style={},
		styleCAB/.style={},	styleCCAB/.style={},	styleCABA/.style={},	styleCABB/.style={},
	}
}
\newcommand\autocataExEnd{%
	\end{tikzpicture}%
}
\newcommand\autocataEx{%
	\matrix[row sep=20, column sep=20, matrixInnerSep] {
		\node[hnode] (E) {$E$};\& \node[hnode] (D) {$D$};		\\
		\node[hnode, styleB] (B) {$B$};	\& \node[hedge, styleBC, label=below:{\labBC}] (BC) {};	\& \node[hnode, styleC] (C) {$C$};	\\
		\node[hnode, styleA] (A) {$A$};\&\& \node[hedge, styleCAB, label=below:{\labCAB}] (CAB) {};		\\
	};
	
	\draw[edge, styleAB] (A.90+\isomerizationOffset) to node[auto, every label] {\labAB} (B.-90-\isomerizationOffset);
	\draw[edge, styleBA] (B.-90+\isomerizationOffset) to node[auto, every label] {\labBA} (A.90-\isomerizationOffset);
	\draw[edge] (E) to node[auto, every label] {\labED} (D);
	
	\draw[edge, styleBBC] (B) to (BC);\draw[edge, styleBCB] (BC) to (C);
	\multiedge{D}{BC}{2}
	\draw[edge, styleCCAB] (C) to (CAB);\draw[edge, styleCABA] (CAB) to (A);\draw[edge, styleCABB] (CAB) to (B);
}
\newcommand\autocataExRestricted{%
	\matrix[row sep=20, column sep=20, ampersand replacement=\&] {
		\node[hnode] (B) {$B$};	\& \node[hedge, label=below:{\labBC}] (BC) {};	\& \node[hnode] (C) {$C$};	\\
		\&\& \node[hedge, label=below:{\labCAB}] (CAB) {};		\\
	};

	\draw[edge, styleBBC] (B) to (BC);
	\draw[edge, styleBCB] (BC) to (C);
	\multiedge{C}{CAB}{1}\multiedge{CAB}{B}{1}
}
\newcommand\autocataExAllIO{
	\makeIO{A}{180}{\labAI}{\labAO}
	\makeIO{B}{180}{\labBI}{\labBO}
	\makeIO{C}{0}{\labCI}{\labCO}
	\makeIO{D}{0}{\labDI}{\labDO}
	\makeIO{E}{180}{\labEI}{\labEO}
}

\newcommand\autocataExExpPruned{
	\matrix[row sep=16, column sep=16] {
		\node[hxnode,label=above:$E$] (E) {};		\& \node[hxnode,label=above:$D$] (D) {};       \\
		\node[hxnode, label={[overlay]above:$B$}] (B) {};	\& \node[hedge, label=below:{\labBC}] (BC) {}; \& \node[hxnode,label={[overlay]above:$C$}] (C) {};    \\
		\node[hxnode, label=below:$A$] (A) {}; \&\& \node[hedge, label=below:{\labCAB}] (CAB) {};       \\
	};
	
	\makeShortcutEdge{A}{90+\isomerizationOffset}{B}{-90-\isomerizationOffset}{\labAB}
	\makeShortcutEdge{B}{-90+\isomerizationOffset}{A}{90-\isomerizationOffset}{\labBA}
	\makeShortcutEdge{E}{\isomerizationOffset}{D}{180-\isomerizationOffset}{\labED}
	
	\makeEdge{D/-90/\isomerizationOffset, D/-90/-\isomerizationOffset, B/0/0}{BC}{C/180/0}
	\makeEdge{C/-90/0}{CAB}{A/0/0, B/-35/2}
}
\newcommand\autocataExExp{
	\autocataExExpPruned
}
\newcommand\makeIOExp[5][\ioEdgeDist]{{
	\makeIO[#1]{#2}{#3}{#4}{#5}
	\newcommand{\Angle}{#3}
	\newcommand{\Anchor}{#2}
	\newcommand\Offset\isomerizationOffset
	\node[tnode] (t-\Anchor-out-IO) at (\Anchor.\Angle+\Offset) {};
	\node[tnode] (t-\Anchor-in-IO) at (\Anchor.\Angle-\Offset) {};
}}
\newcommand\makeIExp[4][\ioEdgeDist]{{
	\newcommand{\Angle}{#3}
	\newcommand{\Anchor}{#2}
	\newcommand{\Dist}{#1}
	\newcommand{\InText}{#4}
	\newcommand\Offset\isomerizationOffset
	\draw[edge] ($(\Anchor.\Angle-\Offset) + (\Angle:\Dist)$) to node[auto, every label]{\InText} (\Anchor.\Angle-\Offset);
	\node[tnode] (t-\Anchor-in-IO) at (\Anchor.\Angle-\Offset) {};
}}

\newcommand\makeOExpStyle[5][\ioEdgeDist]{{
	\newcommand{\Angle}{#3}
	\newcommand{\Anchor}{#2}
	\newcommand{\Dist}{#1}
	\newcommand{\OutText}{#4}
	\newcommand\Offset\isomerizationOffset
	\draw[edge, #5] (\Anchor.\Angle+\Offset) to node[auto, every label]{\OutText} ($(\Anchor.\Angle+\Offset) + (\Angle:\Dist)$);
	\node[tnode, #5] (t-\Anchor-out-IO) at (\Anchor.\Angle+\Offset) {};
}}
\newcommand\autocataExExpAllIO{
	\makeIOExp{A}{180}{\labAI}{\labAO}
	\makeIOExp{B}{180}{\labBI}{\labBO}
	\makeIOExp{C}{0}{\labCI}{\labCO}
	\makeIOExp{D}{0}{\labDI}{\labDO}
	\makeIOExp{E}{180}{\labEI}{\labEO}
}
\newcommand{\makeShortcutEdge}[5]{{
    \newcommand{\Source}{#1}
    \newcommand{\Target}{#3}
    \newcommand{\SourceAngle}{#2}
    \newcommand{\TargetAngle}{#4}
    \newcommand{\lab}{#5}
    \node[tnode] (t-\Source-out-sc-\Target) at (\Source.\SourceAngle) {};
    \node[tnode] (t-\Target-in-sc-\Source) at (\Target.\TargetAngle) {};
    \draw[edge] (t-\Source-out-sc-\Target) to node[auto, every label] {\lab} (t-\Target-in-sc-\Source);
}}
\newcommand{\makeShortcutEdgeStyle}[6]{{
    \newcommand{\Source}{#1}
    \newcommand{\Target}{#3}
    \newcommand{\SourceAngle}{#2}
    \newcommand{\TargetAngle}{#4}
    \newcommand{\lab}{#5}
    \node[tnode, #6] (t-\Source-out-sc-\Target) at (\Source.\SourceAngle) {};
    \node[tnode, #6] (t-\Target-in-sc-\Source) at (\Target.\TargetAngle) {};
    \draw[edge, #6] (t-\Source-out-sc-\Target) to node[auto, every label] {\lab} (t-\Target-in-sc-\Source);
}}
\newcommand{\makeEdge}[3]{{
    \newcommand{\Sources}{#1}
    \newcommand{\EdgeNode}{#2}
    \newcommand{\Targets}{#3}
    \foreach \Source/\SourceAngle/\Bend in \Sources {
        \node[tnode] (t-\Source-out-\EdgeNode) at (\Source.\SourceAngle) {};
        \draw[edge] (t-\Source-out-\EdgeNode) to [bend right=\Bend] (\EdgeNode);
    }
    \foreach \Target/\TargetAngle/\Bend in \Targets {
        \node[tnode] (t-\Target-in-\EdgeNode) at (\Target.\TargetAngle) {};
        \draw[edge] (\EdgeNode) to [bend right=\Bend] (t-\Target-in-\EdgeNode);
    }
}}
\newcommand{\makeEdgeStyle}[4]{{
    \newcommand{\Sources}{#1}
    \newcommand{\EdgeNode}{#2}
    \newcommand{\Targets}{#3}
    \foreach \Source/\SourceAngle/\Bend in \Sources {
        \node[tnode, #4] (t-\Source-out-\EdgeNode) at (\Source.\SourceAngle) {};
        \draw[edge, #4] (t-\Source-out-\EdgeNode) to [bend right=\Bend] (\EdgeNode);
    }
    \foreach \Target/\TargetAngle/\Bend in \Targets {
        \node[tnode, #4] (t-\Target-in-\EdgeNode) at (\Target.\TargetAngle) {};
        \draw[edge, #4] (\EdgeNode) to [bend right=\Bend] (t-\Target-in-\EdgeNode);
    }
}}

\title{Defining Autocatalysis in Chemical Reaction Networks}

\usepackage{authblk}
\usepackage{marvosym}

\newcommand\email[1]{\texttt{#1}}

\author[1]{Jakob L.\ Andersen}
\author[2]{Christoph Flamm}
\author[1]{Daniel Merkle}
\author[2-6]{Peter F.\ Stadler}

\affil[1]{Department of Mathematics and Computer Science, University of Southern Denmark, Campusvej 55, Odense M - DK-5230, Denmark
	\email{\{jlandersen,daniel\}@imada.sdu.dk}}
\affil[2]{Institute for Theoretical Chemistry, University of Vienna, W{\"a}hringerstra{\ss}e 17, A-1090 Wien, Austria
	\email{xtof@tbi.univie.ac.at}}
\affil[3]{Bioinformatics Group, Department of Computer Science;
  Interdisciplinary Center for Bioinformatics; German Centre for
  Integrative Biodiversity Research (iDiv) Halle-Jena-Leipzig; Competence
  Center for Scalable Data Services and Solutions Dresden-Leipzig; Leipzig
  Research Center for Civilization Diseases; and Centre for Biotechnology
  and Biomedicine, University of Leipzig, H{\"a}rtelstra{\ss}e 16-18,
  D-04107 Leipzig, Germany
  \email{studla@bioinf.uni-leipzig.de}}
\affil[4]{Max Planck Institute for Mathematics in the Sciences, Inselstra{\ss}e 22, D-04103 Leipzig, Germany}
\affil[5]{Facultad de Ciencias, Universidad Nacional de Colombia, Sede Bogot{\'a}, Colombia}
\affil[6]{The Santa Fe Institute, 1399 Hyde Park Rd., Santa Fe, NM 87501, United States}

\date{}

\begin{document}
\maketitle
\clearpage

\begin{abstract}
  Autocatalysis is a deceptively simple concept, referring to the situation
  that a chemical species $X$ catalyzes its own formation. From the
  perspective of chemical kinetics, autocatalysts show a regime of
  super-linear growth.  Given a chemical reaction network, however, it is
  not at all straightforward to identify species that are autocatalytic in
  the sense that there is a sub-network that takes $X$ as input and
  produces more than one copy of $X$ as output. The difficulty arises from
  the need to distinguish autocatalysis e.g.\ from the superposition of
  a cycle that consumes and produces equal amounts of $X$ and a pathway that
  produces $X$. To deal with this issue, a number of competing notions,
  such as exclusive autocatalysis and autocatalytic cycles, have been
  introduced. A closer inspection of concepts and their usage by different
  authors shows, however, that subtle differences in the definitions often
  makes conceptually matching ideas difficult to bring together formally.
  In this contribution we make some of the available approaches comparable
  by translating them into a common formal framework that uses integer
  hyperflows as a basis to study autocatalysis in large chemical reaction
  networks. As an application we investigate the prevalence of
  autocatalysis in metabolic networks.
\end{abstract}


\section*{Introduction}

The idea of autocatalysis is deceptively simple. A chemical reaction is
autocatalytic whenever one of its educts catalyzes its own formation, i.e.,
\begin{equation*}
\ce{(A) + X -> 2X + (W)}
\end{equation*}
where \ce{(A)} and \ce{(W)} denote some sets of extra building material and
waste products, respectively. One of the few autocatalytic reactions that
is of this simple form is the Soai reaction \cite{Soai:95}, an alkylation
of pyrimidine-5-carbaldehyde with diisopropylzinc. Here, each enantiomer of
the product catalyzes only the formation of the same enantiomer.

The concept of autocatalysis goes back to Ostwald \cite{Ostwald:1890}. In
the context of chemical kinetics, autocatalysis refers to a temporary
speed-up of the reaction before it settles down to reach equilibrium, see
e.g.\ \cite{Bissette:13,Schuster:19} for a recent review. In most cases
this leads to characteristic sigmoidal time courses. Autocatalysis may also
be associated with more complex dynamic behavior, such as oscillations.

Maybe the best-known example of an autocatalytic reaction is the hydrolysis
of esters, which is catalyzed by the acid that is one of the reaction
products. Even in this simple case, however, we better understand its
autocatalytic nature as a generic acid catalysis of the cleavage reaction
\begin{equation*}
\ce{H^+ + R-COOR^' + H2O -> R-COOH + HO-R^' + H^+}
\end{equation*}
and the dissociation of the acid
\begin{equation*}
\ce{R-COOH <=> R-COO^- + H^+}
\end{equation*}
Of course, the cleavage reaction itself consists of multiple steps, none of
which is overtly catalytic \cite{Bansagi:17}. The mechanisms by which
\ce{Mn^2+} catalyzes the oxidation of oxalate by permanganate in this
classical example of autocatalysis is much less obvious and can be
explained only by an elaborate network of reactions
\cite{Kovacs:04a,Kovacs:04b}.

\begin{figure}
  \centering
  \includegraphics[width=\columnwidth]{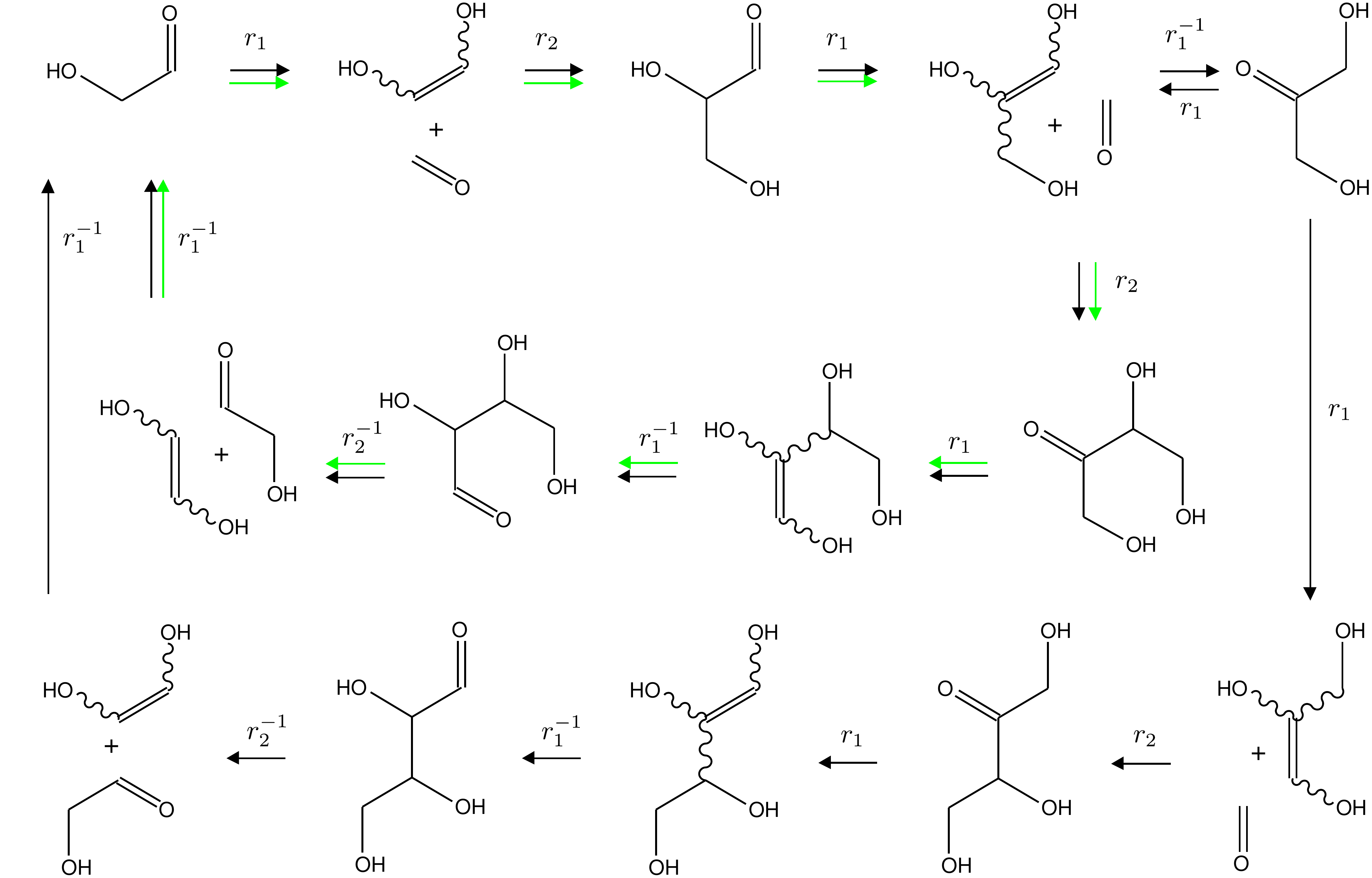}
  \caption{The formose reaction can be understood as a cyclic process
    involving keto-enol tautomerisation ($r_1$ and $r_1^{-1}$),
    aldol-condensation ($r_2$) and reverse aldol reaction $r_2^{-1}$.  The
    inner cycle (green arrows) follows \cite{Benner:10}. The outer cycle is the more
    commonly discussed mechanisms \cite{Breslow:59}. Figure taken from
    \cite{Andersen:14b}.}
  \label{fig:Formose}
\end{figure}

In fact, one of the earliest autocatalytic reactions reported in the
  literature, the Formose reaction \cite{Butlerow:1861}, is a reasonably
well-understood example of ``network autocatalysis''. The simple,
autocatalytic overall reaction
\begin{equation*}
\ce{HOCH2CHO + 2 H2CO -> 2 HOCH2CHO}
\end{equation*}
has been understood as the net effect of the example reaction network
shown in Fig.\ \ref{fig:Formose}. Several well-known oscillating
reactions, including the Belousov-Zhabotinsky (BZ) reaction
\cite{Treindl:97}, are also elaborate examples of network
autocatalysis. In fact, the core of the BZ reaction harbors two
autocatalytic cycles, one feeding on the other in a
\textit{predator-prey} like fashion, resulting in the Lotka-Volterra
type oscillatory dynamics \cite{Ganti:1984}. The Whitesides group
recently designed an autocatalytic network comprising only a few
simple organic compounds that displays oscillatory behavior
\cite{Semenov:2016}. A computational study identified coupled
autocatalytic cycles in the chemical networks of Eschenmoser's
glyoxylate scenario \cite{Andersen:2015}. For the distinction of
catalytic and autocatalytic cycles, see Fig.\ \ref{fig:xtofcyc}.

\begin{figure}
\centering
\subcaptionbox{\label{fig:xtofcyc:cata}}{
\resizebox{0.45\linewidth}{!}{
\begin{tikzpicture}
\foreach \a\l in {1/A, 2/D, 3/C, 4/B}
{ \draw (\a*360/4: 2cm) node[hnode,name=X\l]{\l}; }
\foreach \a\l in {1/4, 2/3, 3/2, 4/1}
{ \draw (\a*360/4+360/8: 2cm) node[hedge,name=R\l]{$r_\l$}; }
\foreach \a\l [evaluate={\ai=int(mod(\a,2));}]
         in {1/a, 2/a, 3/d, 4/d, 5/c, 6/c, 7/b, 8/b}
{
 \ifnum\ai>0
   \draw (\a*360/8-360/16: 3cm) node[hnode,name=Yp\l]{$\l'$};
 \else
   \draw (\a*360/8-360/16: 3cm) node[hnode,name=Y\l]{$\l\phantom{'}$};
 \fi
}
\foreach \x\y in {XA/R1, R1/XB, XB/R2, R2/XC, XC/R3, R3/XD, XD/R4, R4/XA}
{ \draw[edge] (\x) to (\y); }
\foreach \x\y in {Ya/R1, R1/Ypa, Yb/R2, R2/Ypb, Yc/R3, R3/Ypc, Yd/R4, R4/Ypd}
{ \draw[edge] (\x) to (\y); }
\end{tikzpicture}%
}
}
\hfill
\subcaptionbox{\label{fig:xtofcyc:autocata}}{
\resizebox{0.45\linewidth}{!}{
\begin{tikzpicture}
\foreach \a\l in {1/A, 2/D, 3/C, 4/B}
{ \draw (8,0) ++(\a*360/4: 2cm) node[hnode,name=X\l]{\l}; }
\foreach \a\l in {1/4, 2/3, 3/2, 4/1}
{ \draw (8,0) ++(\a*360/4+360/8: 2cm) node[hedge,name=R\l]{$r_\l$}; }
\foreach \a\l [evaluate={\ai=int(mod(\a,2));}]
         in {1/a, 2/a, 3/d, 4/d, 5/c, 6/c, 7/b, 8/b}
{
 \ifnum\ai>0
   \ifnum\a<7
    \draw (8,0) ++(\a*360/8-360/16: 3cm) node[hnode,name=Yp\l]{$\l'$};
   \fi
 \else
   \draw (8,0) ++(\a*360/8-360/16: 3cm) node[hnode,name=Y\l]{$\l\phantom{'}$};
 \fi
}
\foreach \x\y in {XA/R1, R1/XB, XB/R2, R2/XC, XC/R3, R3/XD, XD/R4, R4/XA}
{ \draw[edge] (\x) to (\y); }
\foreach \x\y in {Ya/R1, R1/Ypa, Yb/R2, Yc/R3, R3/Ypc, Yd/R4, R4/Ypd}
{ \draw[edge] (\x) to (\y); }
\draw[edge,sol-cyan,thick] (R2) to[out=135,in=90] (XC);
\draw[edge,sol-cyan,thick] (R2) to[out=135,in=0] (XD);
\draw[edge,sol-cyan,thick] (R2) to[out=135,in=270] (XA);
\draw[edge,sol-cyan,thick] (R2) to[out=135,in=180] (XB);
\end{tikzpicture}%
}
}
\caption{Catalytic and autocatalytic cycle. \subref{fig:xtofcyc:cata} A catalytic cycle
  is shown. A set of educts $\{a,b,c,d\}$ is converted into a set of
  products $\{a',b',c',d'\}$ by a set of species $\{A,B,C,D\}$ which acts as
  as catalysts.  \subref{fig:xtofcyc:autocata} If one of the reactions in the catalytic cycle
  (here $r_2$) additionally produces a copy of the species in the cycle (indicated by
  the green arrows), then the catalytic cycle becomes
  autocatalytic. Green arrows may also represent multi-step reaction
  sequences.}
\label{fig:xtofcyc}
\end{figure}
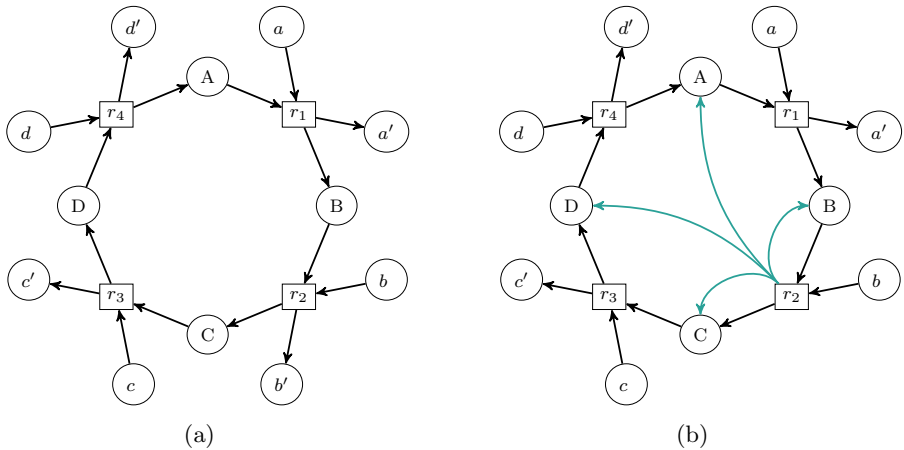

The concept of autocatalysis plays an important role in metabolic networks.
In this context, one frequently speaks of \emph{autocatalytic
  pathways}, which contain reactions that consume some of the pathway's
products. This results in the positive feedback, which in turn explains
their characteristic dynamic behavior \cite{Kun:08,Siami:20}. A
paradigmatic example is glycolysis, which invests two ATP molecules to
later produce four.

Autocatalysis plays a key role in most models of the origin of
life. Replicating entities --- by definition --- are autocatalytic. First
described theoretically by Manfred Eigen \cite{Eigen:71}, it was soon shown
that short nucleic acid templates can be copied, e.g., by ligation of short
fragments without the help of enzymes \cite{Sievers:94}. Alternative
models, such as self-replicating peptides \cite{Lee:96} or lipid aggregates
\cite{Segre:01} follow the same logic. Tibor G{\'a}nti
\cite{Ganti:76,Ganti:76b} early-on emphasized the importance of
autocatalytic cycles. In order to explain the emergence of replicators,
``collectively autocatalytic'' networks of interacting molecules have been
proposed as precursors of replicating polymers \cite{Kauffman:86}. These
chemical reaction networks (CRNs) contain molecules that promote their own
synthesis, forming chemical organizations \cite{Kaleta:06,Benkoe:09a}.  A
distinct concept of ``autocatalytic networks'' refers to interacting
autocatalytic replicators generalizing the hypercycle model of Eigen and
Schuster \cite{Eigen:79,Stadler:92a,BMRStadler:00b}. It describes systems
of self-replicating entities rather than chemical reactions of small
molecules.

A popular mathematical model of autocatalytic reaction networks are the
\emph{Reflexively Autocatalytic Food generated} networks (RAFs) by Steel
and Hordijk \cite{Steel:00,Hordijk:04}. Similar to chemical organizations,
all chemical species in a RAF $\mathcal{R}$ can be produced from the food
or other elements of $\mathcal{R}$ \cite{Hordijk:18}. The model is
mathematically much easier to handle than arbitrary CRNs because one
considers only reactions of the form
\begin{equation*}
  \ce{C +} \sum_{i} s_i \ce{A}_i \ce{->} \sum_{j} s'_j \ce{B}_j \ce{+ C} 
\end{equation*}
That is, every reaction is catalyzed by some of the species. A RAF set
$\mathcal{R}$ thus also contains a sufficient set of catalysts. While RAF
theory is a plausible description, e.g., of ligation networks of simple
polymers \cite{Virgo:16,Liu:18}, it does not seem to be a realistic
description of reaction networks of small molecules. Here, the assumption
that all reactions are catalyzed appears very unrealistic. Unfortunately,
the algorithms for recognizing RAFs \cite{Hordijk:15,Steel:19} do not seem
to generalize to arbitrary networks composed of non-catalyzed
reactions. RAFs are not necessarily meant to model concrete chemical
reactions but rather aggregate transformations.  The RAF formalism
coarse-grains the elementary steps of a catalytic \emph{process} and
replaces them by a single influence arrow. This is a valid abstraction if
enzymes or other large polymeric entities are the catalysts because they
are molecular machines that encapsulate or sequester the individual steps
and thus separate the catalytic process from the rest of the system. It is
not an appropriate approximation for networks of small (prebiotic)
molecules. Here, the intermediates are accessible for alternative
reactions. A specific catalytic influence beyond global effects (such as
changes in pH or ionic strength) in a CRN is itself a chemical reaction. It
remains an open question, therefore, under which conditions a given CRN can
be abstracted into the RAF formalism. We suspect that small molecule CRNs
are not of this type (recent attempts notwithstanding \cite{Xavier:20}),
limiting RAFs to the realm of macromolecular and supramolecular
complexes.

In this contribution, we therefore seek to develop a theory that can be used
to identify autocatalytic structures in a given CRN, i.e., a system of
chemical reaction equations. To this end we first need to introduce a sound
mathematical framework. This is less trivial than it might seem. The
notion of autocatalysis in chemical kinetics is difficult to use in a
network setting since it strongly depends on the actual choice of rate
constants. A natural starting point for a theory of autocatalytic CRNs is
to ask for sub-networks for which the rate constants can be chosen
such that it shows autocatalytic kinetics.  The kinetic criterion, however,
is also not entirely unambiguous as we shall see.

\section*{Towards a Structural Theory of Autocatalysis}

\subsection*{Directed Multi-Hypergraphs}
A CRN consists of a set of molecules $V$ and set of reactions $E$ such
that every $e\in E$ is of the form
\begin{align*}
  \sum_{x\in V} s^+_{xe} x \ce{->} \sum_{x\in V} s^-_{xe} x
\label{eq:CRN}
\end{align*}
Note that we regard all reactions as directed. Reversible reactions
therefore are represented by a separate forward and backward reaction. This
will allows us to use non-negative flows and connects naturally with
graph transformations as a means of generating chemical reactions.

Following the notation of \cite{Andersen:19a}, a CRN is naturally
represented as a directed multi-hypergraph $\mathcal{H} = (V,E)$ where each
hyperedge $e\in E$ is a pair of multisets
\begin{align*} 
e^+:=\mset{x\mid s^+_{xe}>0} \quad\text{and}\quad
e^-:=\mset{x \mid s^-_{xe}>0}
\end{align*}
The notation $\mset{\dots}$ emphasizes that we are dealing with multisets,
where an element can be contained more than once. We write $m_x(\,.\,)$ for
these \emph{multiplicities}, which in our case are given by the
stoichiometric coefficients: $m_x(e^+) = s_{xe}^+$ and
$m_x(e^-) = s_{xe}^-$. We call $e^+$ the \emph{tail} and $e^-$ the
\emph{head} of the directed hyperedge. See Fig.\ \ref{fig:hyperedge} for an
example.

\begin{figure}
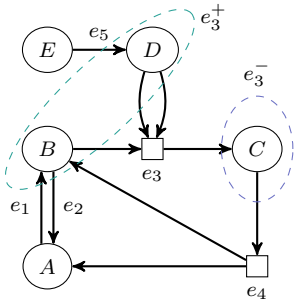

\centering
\autocataExBegin
\autocataEx
\pgfmathanglebetweenpoints{\pgfpointanchor{B}{center}}{\pgfpointanchor{D}{center}}
\let\nodeAngle\pgfmathresult
\node[at=($(B)!0.5!(D)$), draw=sol-cyan, ellipse, dashed, rotate=\nodeAngle, minimum width=95, minimum height=27, label=-10:{$e_3^+$}] {};
\node[at=(C), draw=sol-violet, ellipse, dashed, minimum width=27, minimum height=40, label=above:{$e_3^-$}] {};
\autocataExEnd
\caption{Example of a directed multi-hypergraph $\mathcal{H}$ with vertices
  $V = \{A, B, C, D, E\}$ and hyperedges $E = \{e_1, e_2, e_3, e_4, e_5\}$,
  which we will use as a running example to illustrate concepts.  Parallel
  arrows represent the multiplicity of a vertex in a tail/head multiset.
  For example, the hyperedge $e_3$ is defined as the pair $(e_3^+, e_3^-)$
  with tail $e_3^+ = \mset{B, D}$ (the cyan ellipse)
  and head $e_3^- = \mset{C}$ (the violet ellipse).
  The multiplicities for the tail/head memberships are $m_B(e_3^+) = 1$, $m_D(e_3^+) = 2$ and $m_C(e_3^-) = 1$.
  To reduce clutter we often depict a hyperedge with
  a single tail and head vertex as a lonely arrow without a box (here
  $e_1$, $e_2$, and $e_5$).  The hypergraph represents the reactions \ce{A
    <=> B}, \ce{B + 2 D -> C}, \ce{C -> A + B}, and \ce{E -> D}.  }
\label{fig:hyperedge}
\end{figure}

Every directed multi-hypergraph $\mathcal{H}=(V,E)$
has a faithful representation as a bipartite multi-digraph with vertex set
$V' = V\cup E$ and a multiset of edges
\begin{equation}
  \begin{split}
  E' = &\mset{(v, e)\mid e = (e^+, e^-)\in E, v\in e^+} \\
  \cup &\mset{(e, v)\mid e = (e^+, e^-)\in E, v\in e^-}
  \end{split}
\end{equation}
with multiplicities of arcs given by the stoichiometric coefficients.  We
will refer to this as the \emph{K{\"o}nig representation}
$\koenig(\mathcal{H})$ of $\mathcal{H}$.

In our discussion we also need restrictions of directed hypergraphs to
subsets of vertices and hyperedges in the following way.  For
$V'\subseteq V$ and $E'\subseteq E$ let $\mathcal{H}[V',E']$ be the
directed multi-hypergraph with vertex set $V'$ and the hyperedges
$e' = (e^+ \cap V', e^- \cap V')$ for each $e\in E'$.  In another view,
$\mathcal{H}[V',E']$ is the hypergraph constructed by first taking the
K{\"o}nig representation $\koenig(\mathcal{H})$, selecting the subgraph
induced by $V'\cup E'$, and then reinterpreting it back into a hypergraph.
See Fig.\ \ref{fig:restrictedHypergraph} for an example.

\begin{figure}
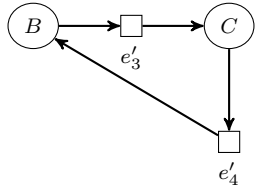

\centering
\autocataExBegin
\def\labBC{$e_3'$}
\def\labCAB{$e_4'$}
\autocataExRestricted
\autocataExEnd
\caption{Example of a restriction of the hypergraph shown in Fig.\
  \ref{fig:hyperedge}.  The shown hypergraph is $\mathcal{H}[V', E']$ with
  $V' = \{B, C\}$ and $E' = \{e_3, e_4\}$.  Note that the original $e_3$
  and $e_4$ have been modified as not all of their tail and head vertices
  are in $V'$.}
\label{fig:restrictedHypergraph}
\end{figure}

The interaction of the CRN $\mathcal{H} = (V, E)$ with its environment is
modeled by ``exchange reactions'' describing the possibility to import or
export/accumulate chemical species.  In general we add these reactions to
every species, and later introduce flow constraints to model specific food
and product sets when relevant.  The exchange reactions are the
\emph{input} edges $E^- = \{e^-_v = (\emptyset, \mset{v}) \mid v\in S\}$
and the \emph{output} edges
$E^+ = \{e^+_v = (\mset{v}, \emptyset) \mid v\in T\}$.  We therefore define
the \emph{extended hypergraph} $\overline{\mathcal{H}}=(V, \overline{E})$
of $\mathcal{H}$ with $\overline{E} = E\cup E^- \cup E^+$. The exchange
reactions appear as sources and sinks in the K{\"o}nig representation
$\koenig(\overline{\mathcal{H}})$. See Fig.\ \ref{fig:extended} for an
example of an extended hypergraph.

\begin{figure}
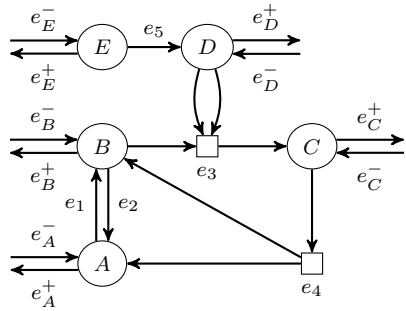

\centering
\autocataExBegin
\autocataEx
\autocataExAllIO
\autocataExEnd
\caption{The extended hypergraph $\overline{\mathcal{H}}$ of the one shown
  in Fig.\ \ref{fig:hyperedge}.  To reduce visual clutter the boxes of the
  IO edges are omitted and only an arc are shown for each of them.}
\label{fig:extended}

\end{figure}

\subsection*{Composite Reactions and Formal Autocatalysis}
On the set of reactions, i.e., hyperedges of a CRN $\mathcal{H}$ we
construct \emph{composite reactions} as integer linear combinations of the
form
\begin{equation}
  \sum_{e\in E} \left( f_e \sum_{x\in e^+} s^+_{xe} x \right)
  \ce{->}
  \sum_{e\in E} \left( f_e \sum_{x\in e^-} s^-_{xe} x \right)
\label{eq:compreact}
\end{equation}
with $f_e \in \mathbb{N}_0$. A composite reaction often contains one or
more species $y$ that appear with the same multiplicity on the both sides,
i.e., $\sum_{e\in E} f_e (s^+_{ye} - s^-_{ye})=0$. These are \emph{formal
  catalysts} for the composite reaction. It is customary to cancel formal
catalysts and to retain in the ``net reaction'' or ``overall reaction''
only the species for which $\sum_{e} (s^+_{ye} - s^-_{ye})f_e \ne 0$.

\begin{definition}
  A composite reaction is \emph{formally autocatalytic for $x$} if it is
  of the form
  \begin{equation*}
  \ce{(A) +} m\ce{x ->} n\ce{x + (W)}
\end{equation*}
for some integers $n>m>0$.
\label{def:FA}
\end{definition}
A CRN is formally autocatalytic if it admits a composite reaction that is
formally autocatalytic for one of its constituent compounds $x$.
Def.\ \ref{def:FA} captures King's notion of autocatalytic sets
\cite{King:78}. It also matches with G{\'a}nti's notion that autocatalysis
is associated with a cycle that eventually feeds a product back as an educt
such that ``after a finite number of turns, each constituent multiplies in
quantity'' \cite{Ganti:76}.

We \emph{conjecture} that it is impossible for a CRN $(V,E)$ to show
kinetic autocatalysis (i.e., non-linear acceleration) if, at least, it is
not formally autocatalytic, since it seems plausible that the presence of a
species $x$ for which the network is formally autocatalytic is a necessary
condition for positive feedback of $x$ on its formation. The notion of
formal autocatalysis implicitly appears in \cite{Virgo:16}, where such
composite reactions are shown to explain superlinear kinetics, i.e.,
autocatalytic behavior, of certain intermediates in a model of a complex
ligation network.

It is easy to see that formal autocatalysis cannot be sufficient. To this
end, consider two (possibly composite) reactions 
\begin{equation}
  \ce{(A) + X -> X + (W)} \qquad 
  \ce{(B) -> X + (U)}
  \label{eq:counter1}
\end{equation}
The first one is a net transformation \ce{(A) -> (W)} catalyzed by \ce{X},
i.e., it does not contribute to the production or degradation of \ce{X},
while the second one is simply a production reaction for \ce{X}. Their sum
is formally autocatalytic for \ce{X} with $m=1$ and $n=2$. Assuming that
the reactions \ce{(A) + X -> X + (W)} and \ce{(B) -> X + (U)} share only
\ce{X}, there clearly is no feedback between them.

In fact, formal autocatalysis is a very weak condition that includes
  reaction mechanisms such as the following 2-step decay of \ce{A}:
\begin{equation}
  \ce{(A) -> 2X}  \qquad
  \ce{X -> (V) + (W)}
  \label{eq:counter2}
\end{equation}
This CRN contains the composite reaction \ce{X + (A) -> 2 X + (V) + (W)},
making \ce{X} formally autocatalytic, even though \ce{X} in no is way
involved in its own production or maintenance.  This emphasizes that the
definition of formal autocatalysis lacks a condition that ties the two
``pathways'' more closely together. These two simple examples naturally
lead to a stricter notion of autocatalysis, ``exclusive autocatalysis'',
where we require that \ce{X} cannot be produced unless \ce{X} is already
present, e.g.\ see \cite{Kun:08}.  However, in order to formalize this idea
properly, we first need to consider integer hyperflows as a way to formalize
the intuitive notion of a pathway.

Before we proceed, we note that network structure alone is certainly
insufficient to imply autocatalysis in the kinetic sense. Even in the
setting of simple ``autocatalytic cycles'', the dynamical behavior depends
crucially on the kinetic parameters \cite{Barenholz:17}.

\subsection*{Integer Hyperflows} 
Pathways, understood as systems of reactions with defined input, are
naturally described mathematically as integer hyperflows
\cite{Andersen:19a}. In this and the following section we introduce some
necessary notation and then explain the connection between integer
hyperflows and the ``algebra'' of reactions in a CRN.

For an extended hypergraph $\mathcal{H} = (V, \overline{E})$, we write
$\delta^{+}_A(v)$ as the set of out-edges from a vertex $v\in V$,
restricted to the edge set $A\subseteq \overline{E}$, i.e.,
$\delta^{+}_A(v)=\{e\in A\mid v\in e^{+}\}$.  Likewise, $\delta^{-}_A(v)$
denotes the restricted set of in-edges incident $v$.
\begin{definition}
  A \emph{hyperflow} on $\overline{\mathcal{H}}$ is a function
  $f\colon \overline{E}\rightarrow \mathbb{R}_0^+$ satisfying, for each
  $v\in V$ the conservation constraint
  \begin{equation}
    \label{eq:masscons}
    \sum_{e\in \delta^+_{\overline{E}}(v)} m_v(e^+)f(e) - 
    \sum_{e\in \delta^-_{\overline{E}}(v)} m_v(e^-)f(e) = 0
  \end{equation}%
\end{definition}
The sum of flow out of each vertex must be the same as the sum of flow into
it.  The concept goes back to \cite{Hoffman:74}. It also naturally appears
in Metabolic Flux Analysis and Flux Balance Analysis: writing
$\mathbf{S}_{ve} := m_v(e^+)-m_v(e^-) = s_{ve}^+-s_{ve}^-$ indeed allows us
to express Eq.\ \eqref{eq:masscons} in matrix notation as
$\mathbf{S} f = 0$.

We write $f_1\le f_2$ if $f_1(e)\le f_2(e)$ holds for all $e\in E$. We
write $f_1< f_2$ if $f_1\le f_2$ and $f_1\ne f_2$.  In contrast we use
$f_1\ll f_2$ if $f_1(e)< f_2(e)$ for every hyperedge $e\in \overline{E}$.
A key property of flows is that linear combinations of flows are again
flows as long as non-negativity is preserved. In particular the difference
of two flows $f_1$ and $f_2$ is still a flow if and only if $f_1-f_2\ge 0$.

In this contribution we shall be interested mostly in \emph{integer
  hyperflows}, which for simplicity we will refer to simply as flows
unless otherwise specified.

For a flow $f$ on $\overline{\mathcal{H}}$ we denote by $S(f)$ and $T(f)$
the \emph{actual} source and target species in a given flow $f$, i.e.,
\begin{equation}
  S(f) = \{v \mid f(e^-_v)>0\} \text{ and }
  T(f) = \{v \mid f(e^+_v)>0\}
\end{equation}
When specifying a model for analysis we may also want to specify \textit{a
  priori} an \emph{allowed} source set $S\subseteq V$ and target set
$T\subseteq V$ in $\overline{\mathcal{H}}$. We refer to the triple
$(\mathcal{H},S,T)$ as the \emph{I/O-constrained extended hypergraph}. In
this situation we are only interested in flows $f$ satisfying
$f(e^-_v) = 0$ for all $v\notin S$ and $f(e^+_v) = 0$ for all $v\not\in T$,
i.e., $S(f)\subseteq S$ and $T(f)\subseteq T$.  In the context of metabolic
networks the sources $S$ are usually given by the food set, and the targets
$T$ are the products that can be removed or accumulated.

\subsection*{Flows for Composite and Net Reactions} 

\begin{figure}
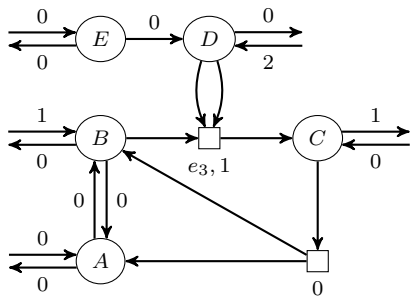

\centering
\autocataExBegin
	\def\labAB{$0$}
	\def\labBA{$0$}
	\def\labBC{$e_3, 1$}
	\def\labCAB{$0$}
	\def\labED{$0$}
	\def\labAI{$0$}\def\labAO{$0$}
	\def\labBI{$1$}\def\labBO{$0$}
	\def\labCI{$0$}\def\labCO{$1$}
	\def\labDI{$2$}\def\labDO{$0$}
	\def\labEI{$0$}\def\labEO{$0$}
\autocataEx
\autocataExAllIO
\autocataExEnd
\caption{The extended hypergraph $\overline{\mathcal{H}}$ from
  Fig.\ \ref{fig:extended}, with $f^{e_3}$ annotated.}
\label{fig:singleEdgeFlow}
\end{figure}

Our next task is to formally connect systems of reactions with flows.
Recall that composite reactions are obtained by ``adding up'' reactions,
i.e., hyperedges. The same can be done for flows. To this end, we associate
each hyperedge $e$ of $\mathcal{H}$, i.e., each reaction in the CRN, with a
flow $f^e$ defined by $f^e(e)=1$, $f^e(e')=0$ for $e'\in E\setminus\{e\}$,
input-flows $f^e(e^-_v)=s_{xe}^-$ for $x\in e^-$ and output-flows
$f^e(e^+_v)=s_{xe}^+$ for $x\in e^+$; all other input- and output-flows are
set to zero. We call $f^e$ the \emph{reaction flow} of $e$. That is, a flow
of $1$ through reaction $e$ requires an input-flow of its educts and an
output-flow of its products in proportions given by the stoichiometric
coefficient.  The reaction flow $f^e$ thus is simply a representation of a
single reaction $e$ in the language of flows. In
Fig.\ \ref{fig:singleEdgeFlow} a reaction flow is shown.

This mathematical construct is useful because it makes it possible to write
the flow $f$ that is associated with a composite reaction (pathway) as a
weighted sum of reaction flows. The multiplicity of a reaction $e$ in
Eq.\ \eqref{eq:compreact} is simply the flow $f(e)$ through $e$ and hence we
have the formal decomposition
\begin{equation}
  f = \sum_{e\in E} f(e) f^{e}
\end{equation}
Recall that in constructing a composite reactions we are only allowed to
add reactions. Thus every compound $v$ comes with an input-flow $f(e^+_v)$
and an output-flow $f(e^-_v)$ that again matches the stoichiometric
coefficients in the composite reactions.

The point of using net reactions, in contrast to using composite reactions,
is that we are allowed to \emph{cancel} intermediates, that is, to remove
an equal number of copies from both the product and the educt side. This
operation can also be formalized in terms of flows. To this end we
introduce the \emph{futile flow} $f^v$ for compound $v$ defined as
$f^v(e^-_v)=f^v(e^+_v)=1$ and $f(e)=0$ for all other reactions
$e\in \overline{E}$. Given a flow $f$, it is easy to see that
$\tilde f= f-c f^v$ is again a valid flow as long as
$c\le \min\{f(e^+_v),f(e^-_v)\}$. That is, we can reduce in $f$ the
input-flow of $v$ and output-flow of $v$ by the same amount as long as we
do not attempt to construct a negative input-flow $\tilde f(e^+_v)$ and or
a negative output-flow $\tilde f(e^-_v)$. In terms of net reaction that
means we may reduce the stoichiometric coefficients of a compound that
appears on both sides by the same amount.

The issue here is that arbitrary canceling of intermediate compounds from
a composite reaction does not necessarily leave us with a net reaction that
will actually take place because we may have canceled essential catalytic
or autocatalytic species. In the flow formalism, however, we can ask which
cancellations are allowed and which are not: We only have to ask whether,
for a given set $S$ of input species and a given set $T$ of output species
there is a flow $f$ with $S(f)\subseteq S$ and $T(f)\subseteq T$, where $S$
and $T$ are subsets of the species on the educt and product side of the
composite reaction. If the answer is yes, we can cancel all intermediate
species $x\in V\setminus (S(f)\cup T(f))$. Correspondingly,
cancellations of $x\in S\cup T$ are not allowed in an I/O constrained
networks $(\mathcal{H}, S, T)$.

Let us write $\supp(f) := \{e\in E\mid f(e)>0\}$ for the set of reactions
(not including I/O hyperedges) that are ``active''. Every flow $f$ can be
associated with a composite reaction, namely the one that consists of all
reactions $e\in\supp(f)$. The stoichiometric coefficients for each $x\in V$
are given by
\begin{equation}
  q^-_x=\hspace*{-1em}\sum_{e\in\supp(f)}\hspace*{-1em} f(e)m_x(e^-)
  \text{ and }
  q^+_x=\hspace*{-1em}\sum_{e\in\supp(f)}\hspace*{-1em} f(e)m_x(e^+)
\end{equation}
Since there is neither an input-flow nor an output-flow for
$x\in V\setminus (S(f)\cup T(f))$, we can conclude immediately that
stoichiometric coefficients of $x$ as an educt, $q^-_x$, and as a product,
$q^+_x$, must be the same. We summarize this discussion as
\begin{lemma}
  There is a flow $f$ on the I/O-constrained network $(\mathcal{H}, S, T)$
  if and only if there is a composite reaction
  $\sum q^-_x x \ce{->} \sum q^+_x x$.  Moreover, in this case
  its stoichiometric coefficients satisfy $q^-_x=q^+_x$ for all
  $x\in V\setminus (S(f)\cup T(f))$.
  \label{lem:flow<->compreact}
\end{lemma} 
In summary, therefore, we can associate a flow with every composite
reaction and \emph{vice versa}. An advantage of the flow framework is that
it links to a convenient computational paradigm. ``Flow queries'',
i.e., the question whether there exists a flow with prescribed properties,
are naturally phrased as (integer) linear programs, and thus can be
answered by generic solvers, see e.g.\ \cite{Andersen:19a} for a more
detailed discussion.

\subsection*{Formally Autocatalytic Flows}

We next link the flow formalism to the notion of formal autocatalysis
introduced in Def.\ \ref{def:FA}. The following statement is a direct
consequence of Lemma \ref{lem:flow<->compreact}, noting that an
(auto)catalytic species necessarily must be contained in both $S(f)$ and
$T(f)$.
\begin{lemma}
  There is a formally autocatalytic compound reaction for $x$ if and
  only if there is a flow $f$ on $\overline{\mathcal{H}}$ such that
  \begin{equation}
    0 < f(e_x^-) < f(e_x^+) 		
    \label{eq:flow-fa}
  \end{equation}
  \label{lem:fauto}
\end{lemma}
In a practical setting we may additionally I/O-constrain $\mathcal{H}$ with
specific source and target sets $S$ and $T$.  The condition matches the
definition of ``overall autocatalysis'' in, e.g.,
\cite{Andersen:19a}. Naturally, we are interested in minimal formally
autocatalytic flows $f$, i.e., those that do not contain a ``smaller''
formally autocatalytic flow $f_1$.

The notion of ``smaller'' in this context deserves some consideration. It
could mean either $\supp(f_1)\subsetneq\supp(f)$ or $f_1<f$. For not
necessarily integer flows, is its well known that the existence of a flow
$f_1$ with $\supp(f_1)\subsetneq\supp(f)$ is equivalent to the existence of
a flow $f_2$ with $f_2<f$ that is not proportional to $f$.  Analogously,
there is an integer flow $f_1$ with $\supp(f_1)\subsetneq\supp(f)$ if and
only if there is an integer flow $f_2$ that is not proportional to $f$ and
an integer $a\ge1$ such that $f_2<a f$. This suggests to think of
``smaller'' flows as those that have the smaller support. Support
minimality features prominently with Extremal Flux Modes
\cite{SchusterS:94,Klamt:17} and has been discussed in detail in this
context.

\subsection*{Exclusive Autocatalysis}
As noted above, the Def.\ \ref{def:FA} and its counterpart in terms of
flows on $\overline{\mathcal{H}}$, Eq.\ \eqref{eq:flow-fa}, are not
satisfactory because parallel reactions such as Eq.\ \eqref{eq:counter1}
and even degradation pathways such Eq.\ \eqref{eq:counter2} are formally
autocatalytic. The most straightforward, but crude way of handling this
shortcoming in the definition is to require, in addition, that an
autocatalytic species $x$ cannot be produced from within the network
unless a minute amount is already present at the outset. In other words,
the network under consideration does not contain a pathway that produces
$x$ in a non-autocatalytic manner from the same food set. This concept
matches the intuition of autocatalysis, e.g., in \cite{Kun:08}, and was
used as a component in \cite{Andersen:19a}.  In the language of flows we
can formalize it as follows:
\begin{definition}
  A species $x$ is \emph{exclusively autocatalytic} in an
  I/O-constrained network $(\mathcal{H}, S, T)$ if there is a flow $f$ such
  that (i) $x$ is formally autocatalytic in $f$ and (ii) there is no
  flow $f_1$ in $(\mathcal{H}, S, V)$ with $f_1(e_x^-)=0$ and
  $f_1(e_x^+)>0$.
  \label{def:obauto}
\end{definition}
Exclusive autocatalysis is a quite strict requirement: if $x$ in any way
can be produced from the sources, without regard to the sinks, it is
disqualified from being exclusively autocatalytic.  Condition (ii) thus
boils down to a simple reachability question in $\mathcal{H}$.  In general,
for a given set of starting materials (``food set'') $F\subseteq V$ and a
set $E'\subseteq E$ of reactions the \emph{scope} \cite{Handorf:05} -- or
the \emph{closure} in the language of chemical organizations
\cite{Kaleta:06,Benkoe:09a} $c(F, E')$ is constructed recursively as
$c(F, E')=\bigcup_i Q_i$, where $Q_0=S$ and, for $i\ge 1$,
\begin{equation}
  Q_i = \bigcup\{ e^-\mid e\in E' \text{ and } e^+\subseteq Q_{i-1}\}
\end{equation}
is the set of a product compounds that can be produced by reactions (in
$E'$) whose educts are available in the previous step $Q_{i-1}$.  An
equivalent way to define $c(F, E')$ is to require $F\in c(F, E')$, and then
for all edges $e\subseteq E'$ if all tail vertices are included,
$e^+\subseteq c(F, E')$, then all head vertices are as well,
$e^-\in c(F, E')$.  Condition (ii) of Def.\ \ref{def:obauto} can thus be
expressed as $x \not\in c(S\setminus \{x\}, E)$.

Def.\ \ref{def:obauto} formalizes a very strict interpretation of the idea
that $x$ cannot be produced unless it is present to seed to its own
production. Condition (ii) is independent of the candidate flow $f$ and
pertains to the complete molecule set $V$ as target set.  As an object of
future study there are several meaningful, less restrictive variations of the
definition, for example:
\begin{enumerate}
\item $f_1$ is found in $(\mathcal{H}, S(f), V)$, allowing all edges $E$,
\item $f_1$ is found in $(\mathcal{H}, S, T)$, allowing all edges $E$,
\item $f_1$ is found in $(\mathcal{H}, S(f), T(f))$, allowing all edges $E$,
\item $f_1$ is found in $(\mathcal{H}, S(f), T(f))$, but allowing only
  edges from $\supp(f)$,
\end{enumerate}
While the first of these variants also can be phrased as a reachability
problem, the others are non-trivial hyperflow queries due to the constraint
on the output flow to a subset of vertices. The last variant can be
interpreted as a question on whether $x$ can be canceled from the educt
side of the composite reaction defined by the formally autocatalytic flow
$f$. This condition therefore is in a sense concerned with the
connectedness of the formally autocatalytic flow $f$. All these concepts of
exclusive or ``obligatory'' autocatalysis are very
restrictive as far as alternative routes are concerned, while the idea of an
underlying autocatalytic cycle is implicit at best.

\subsection*{Autocatalytic Cycles {\itshape{sensu}} Barenholz {\itshape{et
      al.}} (2017)}

Several authors have formalized autocatalysis in terms of the algebraic
properties of the stoichiometric matrix $\mathbf{S}$. In this and the
following section we review two definitions and fit them into the
mathematical framework outlined above, and thus translating them into
constraints on flows on the extended hypergraph
$\overline{\mathcal{H}} = (V, \overline{E})$.

In \cite{Barenholz:17} an autocatalytic cycle is defined as a pair $(M,R)$
of metabolites $M\subseteq V$ and $R\subseteq E$ such that the restriction
$\mathbf{S^*}$ to rows $M$ and columns $R$ satisfies the following 
conditions:
\begin{itemize}
\item[(o)] $R$ contains no reversible pair of reactions.
\item[(i)] For every $x\in M$ there is $e_1,e_2\in R$ with $s_{xe_1}>0$ and
  $s_{xe_2}<0$, and \\
  for every $e\in R$ there is $x_1,x_2\in M$ with $s_{x_1e}>0$ and
  $s_{x_2e}<0$.
\item[(ii)] There is a strictly positive integer vector
  $w\in \mathbb{N}^{|R|}$, $w\gg 0$ such that $\mathbf{S^*} w > 0$,
\item[(iii)] There is no vector $w'>0$ with at least one $e\in R$ for which
  $w'_e=0$ such that $\mathbf{S^*}w>0$.
\end{itemize}
Using the fact that we can express composite reactions as reaction flows we
can rewrite condition (ii) in the flow form as
\begin{itemize}
\item[(ii')] There is a flow $f$ on $\overline{\mathcal{H}}$ such that
  $f(e_v^+)\ge f(e_v^-)>0$ for all $v\in M$ and $f(e)=0$ for all
  $e\in E\setminus R$ and $f(e_v^+)>f(e_v^-)$ for at least one $v\in M$.
\end{itemize}
The first part of condition (i) is equivalent to $x\in M$ appearing on
both sides of the composite reaction, and thus $f(e_v^+),f(e_v^-)>0$ in
the corresponding flow. Thus $M$ consists only of species that are
catalytic ($f(e_v^+)=f(e_v^-)$) or autocatalytic for $v$.  The second
condition constrains $M$ to contain at least one educt and one product of
every $e\in\supp(f)$.

\begin{definition}
  Let $f$ be a flow on $\overline{\mathcal{H}}$. A subset $M\subseteq V$ is
  a \emph{Milo set} for $f$ if
  \begin{itemize}
  \item[\AX{(M1)}] $v\in M$ implies $f(e_v^+)\ge f(e_v^-)>0$,
  \item[\AX{(M2)}] there is $v\in M$ such that $f(e_v^+)>f(e_v^-)$,
  \item[\AX{(M3)}] for all $v\in M$ there is $e'\in E$ with $f(e')>0$ and $v\in
    e'^-$ and $e''\in E$ with $f(e'')>0$ and $v\in e''^+$, and
  \item[\AX{(M4)}] for every $e\in E$ with $f(e)>0$ holds
    $M\cap e^-\ne\emptyset$ and $M\cap e^+\ne\emptyset$.
  \end{itemize}
  A flow $f$ with a Milo set is a \emph{Milo flow}.
  \label{def:Milof}
\end{definition}
Note that \AX{(M2)} implies that a Milo set is non-empty. Furthermore, if
$f$ is a Milo flow, then the Milo set satisfies $M\subseteq S(f)\cap T(f)$.
\begin{lemma}
  If $f$ is a Milo flow on $\overline{\mathcal{H}}$ then $f$ is formally
  autocatalytic for at least one $x\in M$. 
\end{lemma}
\begin{proof}
  By \AX{(M1)}, $f(e^-_v)>0$ for all $v\in M$. Thus \AX{(M2)} implies that
  there is $v\in M$ with $f(e^+_v)>f(e^-_v)>0$, i.e., $v$ is formally
  autocatalytic according to Lemma \ref{lem:fauto}.
\end{proof}

\begin{figure}
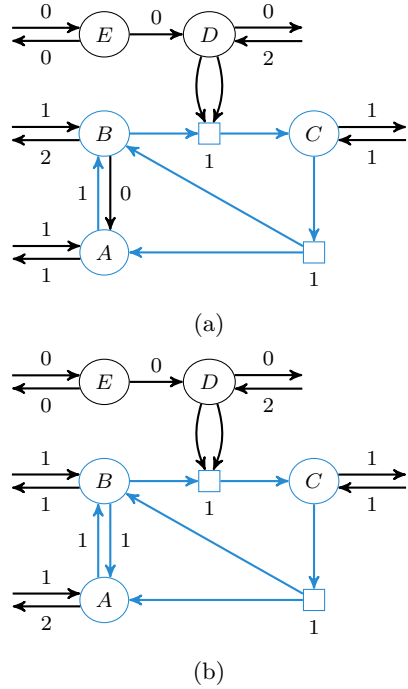

\centering
\subcaptionbox{\label{fig:sub:miloFlow}}{%
\autocataExBegin
	\def\labAB{$1$}
	\def\labBA{$0$}
	\def\labBC{$1$}
	\def\labCAB{$1$}
	\def\labED{$0$}
	\def\labAI{$1$}\def\labAO{$1$}
	\def\labBI{$1$}\def\labBO{$2$}
	\def\labCI{$1$}\def\labCO{$1$}
	\def\labDI{$2$}\def\labDO{$0$}
	\def\labEI{$0$}\def\labEO{$0$}
	\tikzset{
		styleA/.style={draw=sol-blue},
		styleB/.style={styleA},
		styleC/.style={styleA},
		styleAB/.style={styleA},
		styleBC/.style={styleA},		styleBBC/.style={styleA},	styleBCB/.style={styleA},
		styleCAB/.style={styleA},	styleCCAB/.style={styleA},	styleCABA/.style={styleA},	styleCABB/.style={styleA},
	}
\autocataEx
\autocataExAllIO
\autocataExEnd
}

\subcaptionbox{\label{fig:sub:miloFlowReverse}}{%
\autocataExBegin
	\def\labAB{$1$}
	\def\labBA{$1$}
	\def\labBC{$1$}
	\def\labCAB{$1$}
	\def\labED{$0$}
	\def\labAI{$1$}\def\labAO{$2$}
	\def\labBI{$1$}\def\labBO{$1$}
	\def\labCI{$1$}\def\labCO{$1$}
	\def\labDI{$2$}\def\labDO{$0$}
	\def\labEI{$0$}\def\labEO{$0$}
	\tikzset{
		styleA/.style={draw=sol-blue},
		styleB/.style={styleA},
		styleC/.style={styleA},
		styleAB/.style={styleA},		styleBA/.style={styleA},
		styleBC/.style={styleA},		styleBBC/.style={styleA},	styleBCB/.style={styleA},
		styleCAB/.style={styleA},	styleCCAB/.style={styleA},	styleCABA/.style={styleA},	styleCABB/.style={styleA},
	}
\autocataEx
\autocataExAllIO
\autocataExEnd
}

\caption{Two examples of Milo flows on the same
  CRN. \subref{fig:sub:miloFlow} The Milo set $M$ and the supporting
  reactions $\supp(f)$, i.e., restricted network $\mathcal{H}[M, \supp(f)]$
  are highlighted in blue.  The vertex $B$ is formally
  autocatalytic. \subref{fig:sub:miloFlowReverse} Another Milo flow,
  containing a pair of reversible reactions. The equivalent flow $f'$
  obtained by removing $f(AB)=f(\overline{AB})$ is no longer a Milo flow
  because $A$ has no outgoing reaction left in $\supp(f')$. Thus, $f$ is a
  minimal Milo flow. However, $(f,M)$ is not an autocatalytic cycle in the
  sense of Barenholz \cite{Barenholz:17} since the flow $f'$ is a forbidden
  flow according to Def.\ \ref{def:Barenholz}. Since $f'$ itself is not a Milo
  flow, $(f', M)$ is also not an autocatalytic cycle.}
\end{figure}

For a Milo flow $f$, consider the the K{\"o}nig graph
$G:=\koenig(\mathcal{H}[M,\supp(f)])$ of its restriction to the Milo set of
$f$. By \AX{(M3)}, $G$ has no source or sink vertices, i.e., every vertex
of a Milo set is contained in a cycle of $G$.

So far, we have not used conditions (o) and (iii). 
\begin{definition}
  A Milo flow $f$ with Milo set $M$ forms an \emph{autocatalytic cycle}
  $(f,M)$ (\emph{sensu} Barenholz \emph{et al.}, 2007) if there is no flow
  $f_1$ with $\supp(f_1)\subsetneq \supp(f)$ that satisfies \AX{(M1)} and
  \AX{(M2)}.
  \label{def:Barenholz}
\end{definition}
The flow $f$ of an autocatalytic cycle in the sense of Def.\
\ref{def:Barenholz} does not contain a pair of reactions that form a
reversible pair $e$, $\bar e$. If $f$ contains such a reaction, consider
the flow
\begin{equation}
  \begin{split}
    f_1 =& \sum_{e'\in\supp{f}\setminus\{e,\bar e\}} f(e)f^e\\
    & + \min(f(e)-f(\bar e),0) f^e \\
    & + \min(f(\bar e)-f(e),0) f^{\bar e}
  \end{split}
\end{equation}
By construction, $f_1$ coincides with $f$ on $\supp(f)\setminus\{e,\bar e\}$,
has positive input-flow and output-flow, and satisfies that
$f_1(e_v^+)-f_1(e_v^-) = f(e_v^+)-f(e_v^-)$ for every $v\in M$.  Since
$f_1(e)=0$ or $f_1(\bar e)=0$, it is a forbidden flow according to Def.\ \ref{def:Barenholz}.  Thus (o) is in fact a consequence of (iii).

The forbidden flow $f_1$ in Def.\ \ref{def:Barenholz} is a very strong
condition. In particular \cite{Barenholz:17} states (without proof) that
the K{\"o}nig graph of every autocatalytic cycle is strongly connected. At
this point there is no formal proof for this statement, however.

The class of forbidden flows $f_1$ in Def.\ \ref{def:Barenholz} is larger
than Milo flows since $f_1$ is not restricted to flows with inputs and
output from within the set $M$.  Defining a minimal Milo flow to be one for
which there is no Milo flow $f'$ with $\supp(f')\subseteq \supp(f)$, we
observe that every autocatalytic cycle is a minimal Milo flow.  The
converse, however, is not necessarily true, as shown by the example in
Fig.\ \ref{fig:sub:miloFlowReverse}.  It remains an open question whether
all minimal Milo flows are also strongly
connected. Fig.\ \ref{fig:sub:miloFlowReverse} also shows that there are
strongly connected Milo flows that are not autocatalytic cycles
in the sense of Barenholz et al.

\subsection*{Autocatalytic Cores {\itshape{sensu}} Blokhuis {\itshape{et
      al.}} (2020)}

The key concept in \cite{Nghe:20} are submatrices $\mathbf{S^*}$ of the
stoichiometric matrix that are autonomous and productive in the following
sense:
\begin{itemize}
\item[(i)] $\mathbf{S^*}$ is productive if there is a $u\gg 0$ such that
  $\mathbf{S^*}u\gg 0$
\item[(ii)] For every column $e$ of $\mathbf{S^*}$ there are rows
  $v'$ and $v''$ such that $\mathbf{S^*}_{ev'}<0$ and $\mathbf{S^*}_{ev''}>0$.
\end{itemize}
A \emph{autocatalytic core} is defined as a minimal submatrix $\mathbf{S^*}$
of $\mathbf{S}$ with these properties. Prop.~1 in \cite{Nghe:20} shows that
in an autocatalytic core, every species $x$ appears both as a substrate and
as a product. This concept can be rephrased in terms of flows in a manner
that emphasizes its relationship with \cite{Barenholz:17}.
\begin{definition}
  Let $f$ be a flow on $\overline{\mathcal{H}}$. A subset
  $N\subseteq V$ is a \emph{Nghe set} for $f$ if it satisfies 
  \begin{itemize}
  \item[\AX{(N1)}] $v\in M$ implies $f(e_v^+)>f(e_v^-)>0$
  \end{itemize}
  and conditions \AX{(M3)} and \AX{(M4)} of Def.\ \ref{def:Milof}. A flow
  $f$ with a Nghe set $N\ne\emptyset$ is a \emph{Nghe flow}.
  \label{def:Nghef}
\end{definition}
From Def.\ \ref{def:Nghef} we immediately see that every Nghe flow is also a
Milo flow with $M=N$ since \AX{(N1)} obviously implies \AX{(M1)} and
\AX{(M2)}. Thus catalytic cores are Milo flows. In Fig.\ \ref{fig:ngheFlow}
an example of an Nghe flow is shown.

\begin{figure}
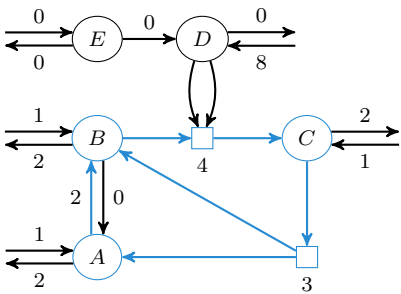

\centering
\autocataExBegin
	\def\labAB{$2$}
	\def\labBA{$0$}
	\def\labBC{$4$}
	\def\labCAB{$3$}
	\def\labED{$0$}
	\def\labAI{$1$}\def\labAO{$2$}
	\def\labBI{$1$}\def\labBO{$2$}
	\def\labCI{$1$}\def\labCO{$2$}
	\def\labDI{$8$}\def\labDO{$0$}
	\def\labEI{$0$}\def\labEO{$0$}
	\tikzset{
		styleA/.style={draw=sol-blue},
		styleB/.style={styleA},
		styleC/.style={styleA},
		styleAB/.style={styleA},
		styleBC/.style={styleA},		styleBBC/.style={styleA},	styleBCB/.style={styleA},
		styleCAB/.style={styleA},	styleCCAB/.style={styleA},	styleCABA/.style={styleA},	styleCABB/.style={styleA},
	}
\autocataEx
\autocataExAllIO
\autocataExEnd
\caption{Example of an Nghe flow,
with the defining Nghe set and supporting reactions highlighted in blue,
i.e., the restricted network $\mathcal{H}[M, \supp(f)]$.
All vertices in the Nghe set, $A$, $B$, and $C$ are formally autocatalytic.
}
\label{fig:ngheFlow}
\end{figure}

Condition \AX{(N1)} appears very restrictive. It will be of immediate
interest, therefore, to better understand under which conditions a Milo
flow contains a Nghe flow in the sense that for a Milo flow $(f,M)$ there
is a Nghe flow $(f_1,N)$ with $N\subseteq M$ and
$\supp(f_1,N)\subseteq\supp(f,M)$. The relationships between Milo and Nghe flows
deserve attention in future work. Similarly, the connections between
autocatalytic cycles \emph{sensu} Barenholz and autocatalytic cores will be
of interest.

Proposition~2 of \cite{Nghe:20} shows that autocatalytic cores $f$ are very
restricted structures: it is ``square'', i.e., $|N|=|\supp(f)|$, every
$x\in N$ is ``the solitary substrate of a reaction, and is substrate for
this reaction only''. Proposition~4 of \cite{Nghe:20}, furthermore, states
that every autocatalytic core is strongly connected. Thus, strongly
connected Nghe flows seem to be interesting objects to study in their own
right.

The work of Nghe \cite{Nghe:20} shows that minimal autocatalytic cores have
an essentially geometric characterization that can be expressed largely in
terms of the K{\"o}nig graph $K:=\koenig(\mathcal{H}[N,\supp(f)])$ of a
minimal Nghe flow.  In essence they can be understood as ``cycles with
ears'' comprising a simple cycle in $K$ augmented by either ``short cut
reactions'' or a path leading from some starting vertex $x$ in the cycle
back to an end-vertex $y$ on the cycle without intersecting the cycle in
its interior. In \cite{Nghe:20}, additional algebraic and minimality
conditions are required for a complete characterization of minimal
autocatalytic cores. This geometric structure suggests to search for
hyperflows whose Milo or Nghe sets have cycles or ears as their
K{\"o}nig graphs.

\section*{Mechanistically Simple Flow Solutions} 
The notion of ``autocatalytic cycles'' and in particular the idea of
``going around a cycle'' to produce additional copies of autocatalytic
compounds \cite{Ganti:76} suggests a definite temporal order in which
molecules ``flow'' through the reactions. This matches the chemist's
concept of a \emph{mechanism} as a sequence of reactions.  Condition
\AX{(M3)} for Milo and Nghe flows addresses this concern to some extent by
requiring input-flow and output-flow for every vertex in the distinguished
set, thus ensuring that a cycle exists in the K{\"o}nig graph of the
support.  On the other hand, minimal Milo and Nghe flows are rather
restrictive in their input/output conditions requiring all vertices in the
distinguished set to be a source and a target. The associated concepts of
autocatalytic cycles or cores, furthermore, ban pairs of reversible
reactions to be used.

The basic flow formulation, and the equivalent formulation based on the
stoichiometric matrix, only ensures mass balance, and does not imply any
particular ordering of reactions as such. In the following we recap the
notion of \emph{expanded flows} from \cite{Andersen:19a}, which has
constraints that introduce localized temporal order in the flow model.  It
makes it feasible to keep a predefined source/target specification in terms
of a I/O-constrained CRN $(\mathcal{H},S,T)$, as well as allowing pairs of
reversible reactions.  The model thus serves as a foundation for finding
chemical pathways in general.  In a later section we sketch how the model
can be enriched with more computational expensive constraints that ensure
the cyclicity required for a comprehensive model of structural
autocatalysis.

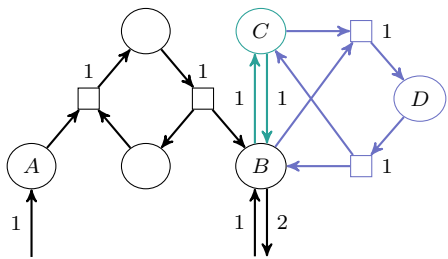
\begin{figure}
\centering
\begin{tikzpicture}
\matrix[matrixInnerSep, row sep=8, column sep=8] {
\&\& \node[hnode] (AB) {\phantom{$A$}};	\&\& \node[hnode, draw=sol-cyan] (C) {$C$};
	\&\&\& \node[hedge, draw=sol-violet, label={[overlay]right:1}] (BCe) {};  \\
\& \node[hedge, label={[overlay]1}] (Ae) {};   \&\&  \node[hedge, label={[overlay]1}] (Be) {};
	\&\&\&\&\& \node[hnode, draw=sol-violet] (D) {$D$}; \\
\node[hnode] (A) {$A$}; \&\& \node[hnode] (BA) {\phantom{$A$}}; \&\& \node[hnode] (B) {$B$};
	\&\&\&  \node[hedge, draw=sol-violet, label={[overlay]right:1}] (De) {};   \\
};
\foreach \s/\t in {A/Ae, Ae/AB, AB/Be, Be/B, Be/BA, BA/Ae} {
	\draw[edge] (\s) to (\t);
};
\draw[edge] ($(A.-90) + (-90:\ioEdgeDist)$) to node[auto, every label] {1} (A.-90);
\makeIO{B}{-90}{1}{2}

\draw[edge, draw=sol-cyan] (B.90+\isomerizationOffset) to node[auto, every label] {1} (C.-90-\isomerizationOffset);
\draw[edge, draw=sol-cyan] (C.-90+\isomerizationOffset) to node[auto, every label] {1} (B.90-\isomerizationOffset);
\foreach \s/\t in {B/BCe, C/BCe, BCe/D, D/De, De/B, De/C} {
	\draw[edge, draw=sol-violet] (\s) to (\t);
};
\end{tikzpicture}
\caption{Example for local reasoning of reaction ordering on a seemingly
  formally autocatalytic flow.  In all interpretations of the flow, the
  violet reactions into and out of \ce{D} forms a futile two-step
  sub-pathway. After removing the flow on these reactions we can then apply
  the same argument to \ce{C}, and then once more on the I/O flow of
  $B$. This leaves the net reaction \ce{A -> B}.}
\label{fig:misleadingFlow}
\end{figure}

As a motivating example, consider the expanded network with flow depicted
in Fig.\ \ref{fig:misleadingFlow}.  The flow, with the net reaction \ce{A +
  B -> 2 B}, is formally autocatalytic, but due to use of reversible
reactions it is equivalent to the simpler reaction \ce{A -> B}.  This can
be established through a step-wise local reasoning:
\begin{enumerate}
\item The only in-flow to \ce{D} is from \ce{B + C -> D} and the only
  out-flow is through the reverse reaction \ce{D -> B + C}.  Any
  ordering of the reactions in the flow will have this two-step futile
  part, and the flow on these reactions can thus be removed.
\item Without the violet part of the network, we can apply the same
  reasoning to vertex \ce{C} with the reactions \ce{B -> C} and \ce{C -> B}.
\item Without both the violet and cyan parts, we can consider \ce{B} with
  its input/output reactions, and decrease the flow by $1$.
\end{enumerate}

\begin{figure}
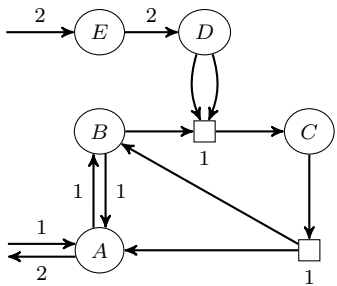

\centering
\autocataExBegin
	\def\labAB{1}\def\labBA{1}
	\def\labED{2}
	\def\labBC{1}
	\def\labCAB{1}
\autocataEx
\makeIO{A}{180}{1}{2}
\draw[edge] ($(E.180) + (180:\ioEdgeDist)$) to node[auto, every label] {2} (E.180);
\autocataExEnd
\caption{Example of a formally autocatalytic flow using pairs of mutually
  reverse reactions. Here, the reactions can be ordered such that no
  two-step futile sub-pathways are present.}
\label{fig:classicalAutocata}
\end{figure}

In this example there is no flow left on reversible reactions, but this is
not the case in general.  Consider the formally autocatalytic flow shown in
Fig.\ \ref{fig:classicalAutocata}, on our running example network.  Here
there are no vertices where we can apply the local temporal reasoning, and
in fact there even exists a partial order for the reactions that have no
pairs of reversible reactions in sequence: {\setlength\multicolsep{1ex}%
\begin{multicols}{2}\noindent
\begin{enumerate}
\item \ce{$\emptyset$ -> E}, twice
\item \ce{E -> D}, twice
\item \ce{$\emptyset$ -> A}
\item \ce{A -> B}
\item \ce{B + 2 D -> C}
\item \ce{C -> A + B}
\item \ce{B -> A}
\item \ce{A ->$\emptyset$}, twice
\end{enumerate}
\end{multicols}}
\par\noindent To consider global ordering one must invoke much
stronger, and computational expensive, formalisms, such as Petri nets that
explicitly ``tracks'' the paths of molecules through the CRN
\cite{Koch:10}.  Note also that in general a pathway may have cycles even
in a fully resolved temporal interpretation.

Fig.\ \ref{fig:classicalAutocata} also shows the requirement of
  mechanistic simplicity can enforce topological constraints. The flows
  obtained by changing the values for the reactions \ce{A -> B} and \ce{B
    -> A} to $0$ or $2$ are no longer chemically simply: In the first case
  it can only be realized by influx of $1$ at \ce{A} that immediately flows
  out again, and in the second case it required the a flow of $1$ reaching
  \ce{B} from \ce{A} is immediately redirected back to \ce{A}.

\subsection*{Expanded Networks and Flows}
To address the need for local routing constraints on flows we introduce the
\emph{expanded} hypergraph \cite{Andersen:19a}.
Given an extended hypergraph $\overline{\mathcal{H}} = (V, \overline{E})$
we expand each vertex into a complete bipartite
graph with vertices corresponding to each in-edge and out-edge.  That is,
for each $v\in V$:
\begin{align*}
  V_v^- &= \{u^-_{ve} \mid \forall e\in \delta^-_{\overline{E}}(v)\} \\
  V_v^+ &= \{u^+_{ve} \mid \forall e\in \delta^+_{\overline{E}}(v)\} \\
  E_v   &= \left\{\left(\mset{u^-}, \mset{u^+}\right)
          \mid u^-\in V_v^-, u^+\in V_v^+\right\}
\end{align*}
The hyperedges $E_v$ all have multiplicity 1 for their tail and head
vertex, and we call these edges the \emph{transit edges} of $v$.  We then
connect the original edges in the natural manner: for each
$e = (e^+, e^-) \in \overline{E}$ the reconnected edge is
$\widetilde{e}= (\widetilde{e}^+, \widetilde{e}^-)$ with
$\widetilde{e}^- = \mset{u^-_{v e}\mid v \in e^-}$ and
$\widetilde{e}^+ = \mset{u^+_{ve}\mid v \in e^+}$.  The multiplicities of
tails and heads correspond to the original multiplicities.  We finally
define the expanded hypergraph
$\widetilde{\mathcal{H}} = (\widetilde{V}, \widetilde{E})$ as
\begin{equation*}
  \widetilde{V} = \bigcup_{v\in V} V_v^- \cup \bigcup_{v\in V} V_v^+
  \quad\textrm{and}\quad
  \widetilde{E} = \bigcup_{v\in V}E_v \cup
  \{\widetilde{e}\mid e\in \overline{E}\}
\end{equation*}
This is again a directed multi-hypergraph where (integer) flows are defined
as usual.  An example of an expanded hypergraph is shown in
Fig.\ \ref{fig:expanded}.

\begin{figure}
\centering
\subcaptionbox{\label{fig:expanded}}{
	\autocataExBegin
		\def\labAB{}\def\labBA{}
		\def\labBC{}\def\labED{}\def\labCAB{}
	\autocataExExp
		\def\labAI{}\def\labAO{}
		\def\labBI{}\def\labBO{}
		\def\labCI{}\def\labCO{}
		\def\labDI{}\def\labDO{}
		\def\labEI{}\def\labEO{}
	\autocataExExpAllIO
	
	\foreach \s/\t/\b in {A-in-IO/A-out-IO, B-in-IO/B-out-IO, C-in-IO/C-out-IO, D-in-IO/D-out-IO, E-in-IO/E-out-IO,
			B-in-sc-A/B-out-sc-A, A-in-sc-B/A-out-sc-B} {
		\draw[tedge, draw=sol-red] (t-\s) to [bend right=-45, looseness=1.2] (t-\t);
	}
	\foreach \s/\t/\b in {A-in-IO/A-out-sc-B/30,    A-in-sc-B/A-out-IO/-40, A-in-CAB/A-out-IO/-5, A-in-CAB/A-out-sc-B/-40,
			B-in-IO/B-out-sc-A/-40, B-in-IO/B-out-BC/-5,    B-in-sc-A/B-out-IO/40, B-in-sc-A/B-out-BC/-40,  B-in-CAB/B-out-sc-A/30, B-in-CAB/B-out-IO/5, B-in-CAB/B-out-BC/-60,
			C-in-BC/C-out-CAB/-45, C-in-BC/C-out-IO/-5,  C-in-IO/C-out-CAB/40,
			D-in-IO/D-out-BC/30,	D-in-sc-E/D-out-BC/-45, D-in-sc-E/D-out-IO/0,
			E-in-IO/E-out-sc-D/0} {
		\draw[tedge] (t-\s) to [bend right=\b, looseness=1] (t-\t);
	}
	\autocataExEnd
}

\subcaptionbox{\label{fig:expandedST}}{
	\autocataExBegin
		\def\labAB{}\def\labBA{}
		\def\labBC{}\def\labED{}\def\labCAB{}
	\autocataExExpPruned
		\def\labAI{}\def\labAO{}
		\def\labBI{}\def\labBO{}
		\def\labCI{}\def\labCO{}
		\def\labDI{}\def\labDO{}
		\def\labEI{}\def\labEO{}
	\makeIOExp{A}{180}{\labAI}{\labAO}
	\makeIExp{E}{180}{\labEI}
	
	\foreach \s/\t/\b in {A-in-IO/A-out-IO, B-in-IO/B-out-IO, C-in-IO/C-out-IO, D-in-IO/D-out-IO, E-in-IO/E-out-IO,
			B-in-sc-A/B-out-sc-A, A-in-sc-B/A-out-sc-B} {
		\draw[tedge, draw=sol-red, draw=none] (t-\s) to [bend right=-45, looseness=1.2] (t-\t);
	}
	\foreach \s/\t/\b in {A-in-IO/A-out-sc-B/30,    A-in-sc-B/A-out-IO/-40, A-in-CAB/A-out-IO/-5, A-in-CAB/A-out-sc-B/-40,
			B-in-sc-A/B-out-BC/-40,  B-in-CAB/B-out-sc-A/30, B-in-CAB/B-out-BC/-60,
			C-in-BC/C-out-CAB/-45, 
			D-in-sc-E/D-out-BC/-45, 
			E-in-IO/E-out-sc-D/0}{
		\draw[tedge] (t-\s) to [bend right=\b, looseness=1] (t-\t);
	}
	\autocataExEnd
}
\caption{Example of expanded hypergraphs.  \subref{fig:expanded} The
  expanded hypergraph $\widetilde{\mathcal{H}}$ of the network from
  Fig.\ \ref{fig:hyperedge}.  The actual vertices are the small black
  circles while the large circles only indicate grouping corresponding to
  the original vertices of $\mathcal{H}$.  The red transit edges are those
  that go between pairs of edges that are mutually reverse of each other in
  the original network.  \subref{fig:expandedST} The effective network that
  can have non-zero flow, when setting the allowed sources to $\{A, E\}$
  and allowed targets $\{A\}$.}
\label{fig:expanded:both}
\end{figure}

For each pair of mutually reverse edges
$e = (e^+, e^-), \bar e = (\bar e^+, \bar e^-)\in \overline{E}$ and a
vertex $v\in e^-$ there is a \emph{futile transit edge}
$t = (u^-_{ve}, u^+_{v\bar e})$. These futile transit edges correspond to
pushing flow back immediately in the opposite direction of a reversible
reaction without first processing the products in a different reaction.  In
Fig.\ \ref{fig:expanded:both} these edges are shown in red. We can now add
constraints on flow reversibility by simply requiring that the flow on the
red edges vanishes, that is, we enforce the constraint $f(t)=0$ for all
futile transit edges in the expanded hypergraph.

As shown in \cite{Andersen:19a} this expanded network model is
computationally not much more difficult to find solutions in than the
original network.  We can thus simply use the expanded network as a
convenient background model for introducing routing constraints.  In
particular, for each flow $\widetilde{f}$ on the expanded network
$\widetilde{\mathcal{H}}$ we can trivially obtain the equivalent flow $f$
on the extended network $\overline{\mathcal{H}}$ simply by contracting the
expanded vertices again.  This leads us to a type of autocatalysis called
\emph{overall autocatalysis} \cite{Andersen:19a}.
\begin{definition}
  A species $x\in V$ is \emph{overall autocatalytic} for a network
  $\mathcal{H} = (V,E)$ if there exist a flow $\widetilde{f}$ on the
  expanded network $\widetilde{\mathcal{H}}$ such that $f(t)=0$ on all
  futile transit edges and the corresponding contracted flow $f$ on the
  extended network $\overline{H}$ satisfies $0 < f(e_x^-) < f(e_x^+)$.
\end{definition}
This model of overall autocatalysis has been implemented using Integer
Linear Programming as an extension of the software package \modAbbr{}
\cite{mod}.  It does not constrain solutions to actually contain a cycle,
but as outlined in the next section, it is already useful in analyzing
chemical systems when coupled with the notion of exclusive autocatalysis
described earlier.

\section*{Autocatalysis in Metabolic Networks} 
\sisetup{group-minimum-digits=4,}

Metabolism as a whole seems to minimize the generation of waste
molecules. Instead, byproducts and waste from one pathway are fed back into
the network as a valuable resource for another. The effect of this
``molecular recycling'', or ``metabolic closure'', is the emergence of
(catalytic) cycles in the reaction network, a necessary precondition for
autocatalysis. Autocatalytic cycles can persist under noisy conditions,
since they can replace mass loss along the cycle. This feature could be
responsible for the inherent robustness of metabolism against fluctuations
\cite{Piedrafita:10}. The embedding of multiple autocatalytic cycles in a
network context results in feedback between cycles, giving rise to a rich
repertoire of dynamic behavior and entry-points for regulation and
control. Autocatalysis therefore plays an important role in metabolic
networks.

Already in 2008, Kun and collaborators \cite{Kun:08} published a search for
obligatory autocatalytic species in large metabolic network models.
Using the RAF framework, autocatalytic sets in the metabolic network of
\emph{E.\ coli} were studied in \cite{Sousa:15}. In order to illustrate the
theoretical considerations in the previous sections we survey overall
autocatalytic molecules in the metabolic networks of five very different
prokaryotes as retrieved from the BiGG database \cite{Schellenberger2010},
see Tab.\ \ref{tab:bigg}. We only give a cursory overview here, a full
investigation of autocatalysis using flows in these networks is
forthcoming.

\begin{table*}
\caption{Overview of the investigated BiGG models.  The original networks
  were simplified to study the capabilities of the cytosol compartment.
  The number of molecules that could be detected to be overall
  autocatalytic, using routing constraints in the expanded network are
  listed under ``\#OA''.  When further applying the strict conditions of
  exclusive autocatalysis (``EA'') we are left with the number of molecules
  listed in the final column.}
\label{tab:bigg}
\sisetup{
	table-format=4,
}
\centering
\footnotesize
\begin{tabular}{@{}ll@{}c@{}SS@{}c@{}SS@{}c@{}SS@{}}
\toprule
&	&\hspace*{2em}& \multicolumn{2}{@{}c@{}}{Original}	&\hspace*{2em}& \multicolumn{2}{@{}c@{}}{Simplified}
&\hspace*{2em}	\\
\cmidrule{4-5}\cmidrule{7-8}	
  Species		& BiGG ID && {$|V|$} & {$|E|$} && {$|V|$} & {$|E|$}
  && {\#OA}	& {\#(OA + EA)}	\\
\midrule
  \textit{Escherichia coli}           & iML1515
        && 1877	& 3005 && 1434 & 2188 && 736 & 580 \\
  \textit{Helicobacter pylori}        & iIT341
        &&  485	&  641 &&  485 &  641 && 176 & 143 \\
  \textit{Methanosarcina barkeri}     & iAF692
        &&  626	&  809 &&  626 &  809 && 154 & 131 \\
  \textit{Mycobacterium tuberculosis} & iEK1008
        &&  969	& 1372 &&  969 & 1372 && 459 & 385 \\
  \textit{Staphylococcus aureus}      & iYS854
        && 1129	& 1587 && 1126 & 1579 && 506 & 368 \\
\bottomrule
\end{tabular}
\end{table*}

The BiGG models contain multiple copies of some molecules representing
the compartments cytosol, periplasm, and the external environment. Here, we
are only interested in the cytosolic metabolism. We therefore merged the
periplasm with the external compartment and removed all reactions without
educts or products in the cytosol. The size of the original and
simplified networks are listed in Tab.\ \ref{tab:bigg}. We then obtained
the I/O-constrained hypergraphs interpreting the external molecules as
source and target compounds. Furthermore, the explicit exchange
pseudo-reactions in the models were converted into
source/product specifications.

A molecule can only be (formally or overall) autocatalytic if it appears
both as an educt and as a product, thus emulating that it may accumulate in
the cell. Fixing a molecule \ce{X} of interest, we construct an expanded
flow model in which we add the condition that \ce{X} is overall
autocatalytic as an additional constraint. In total, this yields \num{4640}
different flow models of which \num{2031} had feasible solutions; see Tab.\
\ref{tab:bigg} for a summary. Since many of the solutions in essence
conform to Eq.\ \eqref{eq:counter1} and thus do not represent autocatalysis
in a chemically meaningful sense, we restricted ourselves to overall
autocatalytic molecules that are also exclusively autocatalytic in the
sense of Kun et al.~\cite{Kun:08}.  That is, if a molecule is reachable
from the sources (without itself), then it is not considered autocatalytic.
This leaves \num{1607} solutions.

The intersection of the five models shares 245 cytosolic molecules, of
which 87 are overall autocatalytic. Only the 37 molecules listed in
Tab.\ \ref{tab:autocataMols} are also exclusively autocatalytic.

\begin{table}
\caption{List of the 37 molecules that are (i) present in all five of the
  investigated models, (ii) overall autocatalytic, and (iii) exclusively
  autocatalytic.}
\label{tab:autocataMols}
\centering
\footnotesize
\begin{tabular}[t]{@{}ll@{}}
\toprule
BiGG ID			& Name					\\
\midrule
\texttt{adp}		& ADP					\\
\texttt{amp}		& AMP					\\
\texttt{atp}		& ATP					\\
\texttt{cdp}		& CDP					\\
\texttt{cmp}		& CMP					\\
\texttt{ctp}		& CTP					\\
\texttt{dudp}		& dUDP					\\
\texttt{dump}		& dUMP				        \\
\texttt{dutp}		& dUTP					\\
\texttt{gdp}		& GDP					\\
\texttt{gmp}		& GMP					\\
\texttt{gtp}		& GTP				        \\
\texttt{udp}		& UDP					\\
\texttt{udpg}		& UDPglucose				\\
\texttt{udpgal}		& UDPgalactose				\\
\texttt{ump}		& UMP					\\
\texttt{utp}		& UTP					\\
\midrule
\texttt{nad}		& NAD					\\
\texttt{nadh}		& NADH					\\
\texttt{nadp}		& NADP					\\
\texttt{nadph}		& NADPH					\\
\midrule
\texttt{10fthf}		& 10-Formyltetrahydrofolate 		\\
\texttt{methf}		& 5,10-Methenyltetrahydrofolate		\\
\texttt{mlthf}		& 5,10-Methylenetetrahydrofolate	\\
\texttt{thf}		& 5,6,7,8-Tetrahydrofolate		\\
\texttt{thdp}		& 2,3,4,5-Tetrahydrodipicolinate	\\
\texttt{23dhdp}		& 2,3-Dihydrodipicolinate		\\
\midrule
\texttt{4pasp}		& 4-Phospho-L-aspartate			\\
\texttt{aspsa}		& L-Aspartate 4-semialdehyde		\\
\texttt{phom}		& O-Phospho-L-homoserine		\\
\texttt{pser\_\_L}	& O-Phospho-L-serine			\\
\texttt{gal1p}		& Alpha-D-Galactose 1-phosphate		\\
\texttt{13dpg}		& 3-Phospho-D-glyceroyl phosphate	\\
\texttt{prpp}		& 5-Phospho-alpha-D-ribose 1-diphosphate\\
\texttt{3php}		& 3-Phosphohydroxypyruvate		\\
\texttt{actp}		& Acetyl phosphate		        \\
\texttt{ppi}		& Diphosphate				\\
\bottomrule
\end{tabular}
\end{table}

This list for the most part comprises the expected ``currencies'' in the
cell, in particular the mono-, di-, and tri-phosphorylated nucleotides, and
the redox cofactors NAD and NADP. This matches the identification of
ATP/ADP as ubiquitous ``obligatory autocatalysts'' in \cite{Kun:08} using a
very different approach. Furthermore, several tetrahydropholate
derivatives, which are essential cofactors in the single carbon metabolism
and two prebiotically relevant amino acids aspartate and serine are on the
list. Interestingly, also the non-proteinogenic amino acid homoserine, an
intermediate in the biosynthesis pathways of the three essential amino
acids methionine, threonine, and isoleucine, as well as
aspartate-semialdehyde, a building block involved in the biosynthesis of the
amino acids lysine and homoserine are present.

\subsection*{Structural Constraints on Autocatalysis}

In the preceding section we have reviewed several ways of formalizing
autocatalysis in terms of integer hyperflows. While the comparison of the
different approaches provides many open question for future research, it
also leaves the impression that none of them already provides a
satisfactory theory. Querying for Milo and Nghe flows, for instance,
requires a very loose definition of sources and sinks, and only partially
includes structural constraints. Overall autocatalysis provides much more
flexibility in the source/sink specification and provides solutions
directly interpretable as chemical pathways.  However, even with routing
constraints on expanded flows the solutions are not guaranteed to have the
cyclic motifs we would expect for ``true'' autocatalysis.

\begin{figure}
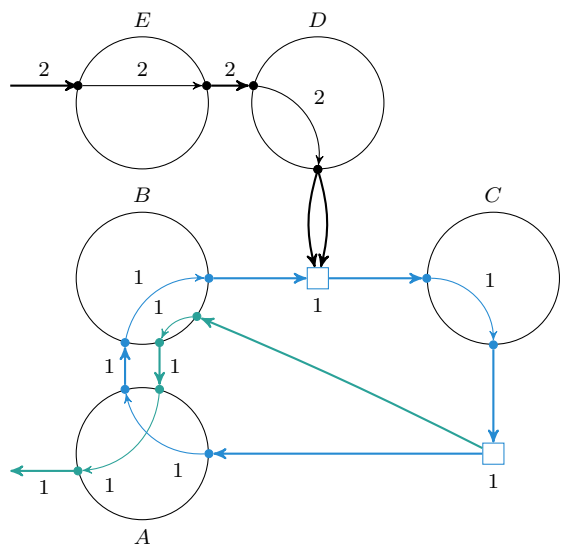

\tikzset{theCycle/.style={color=sol-blue, text=black}}
\tikzset{theEar/.style={color=sol-cyan, text=black}}
\centering
	\autocataExBegin
		\def\labAB{1}\def\labBA{1}
		\def\labED{2}\def\labBC{1}\def\labCAB{1}
	\matrix[row sep=16, column sep=16] {
		\node[hxnode,label=above:$E$] (E) {};		\& \node[hxnode,label=above:$D$] (D) {};       \\
		\node[hxnode, label={[overlay]above:$B$}] (B) {};	\& \node[hedge, theCycle, label=below:{\labBC}] (BC) {}; \& \node[hxnode,label={[overlay]above:$C$}] (C) {};    \\
		\node[hxnode, label=below:$A$] (A) {}; \&\& \node[hedge, theCycle, label=below:{\labCAB}] (CAB) {};       \\
	};
	
	\makeShortcutEdgeStyle{A}{90+\isomerizationOffset}{B}{-90-\isomerizationOffset}{\labAB}{theCycle}
	\makeShortcutEdgeStyle{B}{-90+\isomerizationOffset}{A}{90-\isomerizationOffset}{\labBA}{theEar}
	\makeShortcutEdge{E}{\isomerizationOffset}{D}{180-\isomerizationOffset}{\labED}
	
	\makeEdge{D/-90/\isomerizationOffset, D/-90/-\isomerizationOffset}{BC}{}
	\makeEdgeStyle{B/0/0}{BC}{C/180/0}{theCycle}
	\makeEdgeStyle{C/-90/0}{CAB}{A/0/0}{theCycle}
	\makeEdgeStyle{}{CAB}{B/-35/2}{theEar}
		\def\labAI{}\def\labAO{1}
		\def\labBI{}\def\labBO{}
		\def\labCI{}\def\labCO{}
		\def\labDI{}\def\labDO{}
		\def\labEI{2}\def\labEO{}
	\makeOExpStyle{A}{180}{\labAO}{theEar}
	\makeIExp{E}{180}{\labEI}
	
	\foreach \s/\t/\b in {D-in-sc-E/D-out-BC/-45, E-in-IO/E-out-sc-D/0}{
		\draw[tedge] (t-\s) to [bend right=\b, looseness=1] node[auto, every label] {2} (t-\t);
	}
	\draw[tedge, theEar] (t-A-in-sc-B) to [bend right=-40, looseness=1] node[auto, every label, pos=0.9] {1} (t-A-out-IO);
	\draw[tedge, theCycle] (t-A-in-CAB) to [bend right=-40, looseness=1] node[auto, every label, pos=0.1] {1} (t-A-out-sc-B);
	\draw[tedge, theEar] (t-B-in-CAB) to [bend right=30, looseness=1] node[auto, swap, every label] {1} (t-B-out-sc-A);
	\foreach \s/\t/\b in {B-in-sc-A/B-out-BC/-40, C-in-BC/C-out-CAB/-45} {
		\draw[tedge, theCycle] (t-\s) to [bend right=\b, looseness=1] node[auto, every label] {1} (t-\t);
	}
  \autocataExEnd
  \caption{An expanded flow that contains a catalytic cycle (in
    blue), from which an ``ear'' (in green), produces an additional
    copy of the autocatalytic molecule $A$.}
\label{fig:cycleEar}
\end{figure}

Using the mathematical setup of expanded hypergraphs we can relax the
condition of an autocatalytic vertex to have explicit input, and instead
require that the flow must induce a cycle that goes through any of the
associated vertices in the expanded graph.  More formally, for a flow
$\widetilde{f}$ on the expanded network
$\widetilde{\mathcal{H}} = (\widetilde{V}, \widetilde{E})$, if a vertex
$x\in V$ is to be considered autocatalytic then
$\koenig(\widetilde{\mathcal{H}}[\widetilde{V}, \supp(\widetilde{f})])$
must contain a cycle passing through $x$ (indicated in blue in
  Fig.~\ref{fig:cycleEar}) and an ``ear'' (indicated in green in
  Fig.~\ref{fig:cycleEar}) that branches off the cycle before $x$, rejoins
  the cycle to pass through $x$ and eventually connects to an outflow,
  possibly after additional reactions in parallel with the cycle. This
``ear'' condition may sound deceptively simple, as it directly aligns with
the expectation that the cycle must be productive, but providing a formal
definition requires careful attention.  The cycle condition is
mathematically easy to state, but it is a non-local constraint that may
require a non-trivial computational effort to handle.  We envision that
systems such as the one in Fig.\ \ref{fig:cycleEar} will be a paradigmatic
example of autocatalytic mechanisms. It is worth noting that at least
conceptually this fits with autocatalytic cores of Blokhuis \textit{et
    al.}\ \cite{Nghe:20}.

\section*{Concluding Remarks} 

We cannot claim to have a comprehensive mathematical theory of (structural)
autocatalysis. However, we have a starting point to develop such a theory
and some hints that we can use to guide us into the right direction:
integer hyperflows provide a powerful mathematical framework in which
\emph{some} of the properties of autocatalytic networks can be expressed
very naturally. In addition, it makes the incorporation of stoichiometric
balance conditions very easy and natural.  On the other hand, flows alone
do not seem to be sufficient, since autocatalysis involves pushing material
around in a (generalized) cycle, and thus involves a temporal order of
reactions that -- in general -- is not specified completely by a flow,
which in essence is just a set of reactions. To this end, we have
introduced expanded hypergraphs that encode some of the necessary temporal
ordering. Since flows are by construction a description of a steady state,
we suspect that flows are an inherently incomplete framework, which need to
be complemented by constraints such as the cycle/ear motif sketched in the
previous section that imply temporal order of reactions, i.e., \emph{a
  mechanism}. It remains an interesting mathematical question for
  future research to what extent routing constraints in extended
  hypergraphs imply topological orders of reactions for a given flow.

\bigskip

\fontsize{7}{8}\sffamily

\par\noindent{\bfseries Acknowledgements}

We thank Philippe Nghe for stimulating discussions on autocatalysis during
a joint visit of CERN in October 2019, and for providing the manuscript
``Minimal Autocatalytic Stoichiometries'' to us prior to publication.  This
work was supported in part by the German Federal Ministry of Education and
Research (BMBF) within the project Competence Center for Scalable Data
Services and Solutions (ScaDS) Dresden/Leipzig (BMBF 01IS14014B).
It was also supported by the Independent Research Fund Denmark (DFF-7014-00041.130).

\vspace*{3mm}

\par\noindent{\bfseries Author Contributions}

The authors jointly conceived the study, JLA and PFS developed most of the
mathematical framework, JLA performed the computational analysis of the
metabolic networks. All authors contributed to the manuscript and approved
of its submission. 

\vspace*{3mm}

\bibliographystyle{bmc-mathphys}
\begingroup
\renewcommand{\refname}{\fontsize{7}{8}\bfseries References}
\def\bibfont{\fontsize{7}{8}}
\makeatletter
\def\@biblabel#1{#1.}
\makeatother
\bibliography{kiedro}
\endgroup

\end{document}